\documentclass[10pt]{IEEEtran}
\usepackage{url}
\usepackage{verbatim}
\usepackage{graphics,graphicx}
\usepackage{amsthm,amsmath,amssymb,enumerate}
\usepackage{enumerate}
\usepackage{algorithm}
\usepackage{algorithmic}
\usepackage{color}

\newcommand{\ignore}[1]{}
\newcommand{\ncom}[1]{}

\newtheorem{lemma}{Lemma}
\newtheorem{claim}{Claim}
\newtheorem{theorem}{Theorem}
\newtheorem*{theorem*}{Theorem}

\begin{document}
\title{Kmerlight: fast and accurate $k$-mer abundance estimation}

\author{
    \IEEEauthorblockN{Naveen Sivadasan\IEEEauthorrefmark{1}, Rajgopal Srinivasan\IEEEauthorrefmark{1}, Kshama Goyal\IEEEauthorrefmark{2}}\\
    \IEEEauthorblockA{\IEEEauthorrefmark{1}TCS Innovation Labs Hyderabad, India
    \\{\{naveen, raj\}@atc.tcs.com}}\\
    \IEEEauthorblockA{\IEEEauthorrefmark{2}kshama.goyal@gmail.com}
}

\maketitle



\begin{abstract}
$k$-mers (nucleotide strings of length $k$) form the basis of several algorithms in computational genomics.
In particular, $k$-mer abundance information in sequence data is useful in read error correction, parameter estimation for genome assembly, digital normalization etc. 
We  give a streaming algorithm Kmerlight for computing the $k$-mer abundance histogram from sequence data. Our algorithm is fast and uses very small memory footprint. We provide analytical bounds on the error guarantees of our algorithm.  Kmerlight can efficiently process genome scale and metagenome scale data using standard desktop machines. Few applications of abundance histograms computed by Kmerlight are also shown. We use abundance histogram for de novo estimation of repetitiveness in the genome based on a simple probabilistic model that we propose. We also show estimation of $k$-mer error rate in the sampling using abundance histogram.
Our algorithm can also be used for abundance estimation in a general streaming setting.
The Kmerlight tool is written in C++ and is available for download and use from https://github.com/nsivad/kmerlight.
\end{abstract}

\section{Introduction}
$k$-mers (nucleotide strings of length $k$)	 play a fundamental role in computational genomics. $k$-mers form the basis of many assembly and alignment algorithms. 
Understanding abundance of $k$-mers in sequence reads have applications in read quality estimation, read error correction \cite{ChaissonPevzner2008, chaisson2004, Pevzner2001}, parameter estimation for genome assembly \cite{chikhi2013Genie} and digital normalization \cite{brown2012reference}, etc. 
We consider the problem of computing $k$-mer abundance in a sequence or in a read collection. Specifically, the goal is to compute the count of distinct $k$-mers occurring in the input as well as counts of distinct $k$-mers occurring in the input with given multiplicity (frequency) values. The histogram of such $k$-mer counts for different multiplicity values is referred to as $k$-mer abundance histogram. 

$k$-mer abundance computation has several applications in genome analysis.
In high throughput sequencing, $k$-mer abundance helps in assessing quality of sequence reads and in identifying $k$-mers originating from erroneous reads.
This approach is typically used in spectral alignment techniques for read error detection and correction \cite{kelley2010quake, ChaissonPevzner2008, chaisson2004, Pevzner2001}.
\emph{de Bruijn} graphs of $k$-mers are used in several assembly algorithms. 
Quality of de Bruijn based assemblers crucially depend on the value of $k$-mer size and $k$-mer abundance histogram is helpful in choosing 
appropriate $k$-mer size \cite{chikhi2013Genie}.  
$k$-mer abundance computation also finds application in estimation of $k$-mer error rate in the reads, which helps in understanding sampling error rate \cite{kmerstream}. 
Understanding abundance of  $k$-mers in the genome provides insights into the sequence repetitiveness in the genome. 
Variations in shape of $k$-mer abundance histogram of reads, in particular the positions of peaks in the histogram, have relations to sequencing bias and to the presence of highly polymorphic genomes that contain large number of heterozygous locations in their haplotypes\cite{butler2008allpaths}. 

There are several existing techniques for $k$-mer counting and $k$-mer abundance computation. 
In $k$-mer counting, count of each $k$-mer present in the input is computed.
Algorithms for $k$-mer counting can be used for computing $k$-mer abundance. 
Existing $k$-mer counting tools include  Tallymer \cite{kurtz2008Tally}, Jellyfish \cite{marccais2011jellyfish}, KMC2 \cite{KMC2015}, MSPKmerCounter \cite{li2015mspkmercounter}, DSK \cite{rizk2013dsk}, BFCounter \cite{melsted2011BFCounter}, Khmer \cite{zhang2014khmer}, KAnalyze \cite{audano2014kanalyze}, KmerGenie \cite{chikhi2013Genie}  and Turtle \cite{roy2014turtle}. These tools perform either exact counting or approximate counting. Exact counting techniques rely on large main memory or on a combination of external memory and main memory to handle massive memory requirement. BFCounter uses Bloom filter as a pre-filter to reduce memory requirement. Kmergenie uses random sampling to do approximate counting with reduced memory and compute requirements \cite{chikhi2013Genie}.
It is however known that distinct ($k$-mer) count estimates computed by sampling based approaches are known to have large variance in general unless the sample size is close to the input size \cite{charikar2000}.  Khmer uses CountMin sketch\cite{zhang2014khmer} for approximate counting, for which error margins could be large in general.
KmerStream \cite{kmerstream} is a streaming algorithm for in-memory estimation of the number of distinct $k$-mers and the number of unique $k$-mers. Their approach is an extension of streaming algorithms for estimation of distinct count \cite{barYosef2002}. 
Estimation of $k$-mer counts with larger multiplicities was left as an open problem in \cite{kmerstream}. 


\subsection{Our Contribution}
We present a streaming algorithm Kmerlight to estimate the total number of distinct $k$-mers in the input, denoted as $F_0$, as well as the histogram of total number of $k$-mers with  multiplicity $i$, denoted as $f_i$. 
Our algorithm is an extension of streaming algorithms for count distinct problem \cite{barYosef2002} and for counting unique $k$-mers \cite{kmerstream}. 
Streaming algorithm for computing provable estimates for  $f_i$ for  $i > 1$ was left as an open problem in \cite{kmerstream} and we solve this problem.
Kmerlight is very fast and it uses small in-memory data structures to compute estimates for $F_0$ and $f_i$ values with high accuracy. 
It uses logarithmic space and runs in linear time. 
We also provide analytical bounds on the error margins achieved by Kmerlight. 
To the best of our knowledge, our algorithm is the first streaming algorithm to efficiently compute $F_0$ as well as $f_i$ values with analytical guarantees. 

We conducted several experiments to measure the accuracy and performance of Kmerlight. For instance, with less than 500 MB RAM (and no disk space), Kmerlight achieved 2\% relative error. We provide a multi-threaded C++ implementation of our tool which is available for download and use. 
Kmerlight can be run on a standard desktop machine and it scales well to genome scale and metagenome scale data. 
Resource frugal nature of Kmerlight allows simultaneous computation of abundance histograms with different values of $k$, which for instance is required in parameter estimation for genome assembly. 

We also exhibit few applications of Kmerlight. 
We use Kmerlight for de novo estimation of $k$-mer repetitiveness in the underlying genome from reads. We propose  a simple probabilistic  model for $k$-mer abundance histogram and use it for the estimation. Analyzing $k$-mer repetitiveness  helps in understanding sequence repetitiveness in the genome. We also use Kmerlight for estimation of $k$-mer error rate in the reads and for estimation of the genome size.

Our algorithm can be used for abundance estimation in a general setting.
To the best of our knowledge, our algorithm is the first streaming algorithm to solve this problem using only sublinear space and query time and with analytical bounds. This we believe is of independent theoretical and practical interest. 


\section{Methods}

Let $F_0$ denote the total number of distinct $k$-mers in the input. For $i \ge 1$, let $f_i$ denote the total number of $k$-mers each occurring with multiplicity (frequency) exactly $i$ in the 
input.
The histogram of $f_i$ values constitutes the $k$-mer abundance histogram.

\subsection{ Algorithm }
We present a streaming algorithm Kmerlight to estimate $F_0$ and $f_i$ for $i \ge 1$. 
Our algorithm maintains a `sketch' of the input seen so far. The space used for this sketch is sublinear in the input size. 
Experiments show that the sketch size in the range of  500 MB to 1 GB RAM can provide high accuracy estimates. The sketch is updated upon seeing each $k$-mer in the input. At any stage, estimates for $F_0$ and $f_i$ values are computed from the sketch. 

Our algorithm is an extension of streaming algorithms for count distinct problem \cite{barYosef2002} and for counting unique $k$-mers \cite{kmerstream}. 
We first give a brief overview of the KmerStream algorithm of \cite{kmerstream} for estimation of $F_0$ and $f_1$. Their algorithm is an adaptation of the $F_0$ estimation algorithm in \cite{barYosef2002}. The general idea 
is to perform multi-level sampling of the input $k$-mer stream. 
Each distinct $k$-mer is sampled exclusively by one of the levels and all occurrences of this $k$-mer is assigned to the same level. Sampling is such that level $j$ samples all occurrences of roughly $1/2^j$ fraction of all distinct $k$-mers in the input. 
At each sampling level, an array of counters is maintained.
When a $k$-mer is assigned to a level, it is further hashed to one of the counters in the counter array and the destination counter is incremented by one. 

The number of counters with value zero is used to estimate $F_0$. Estimated value of $F_0$ along with number of counters containing value 1 is used to estimate $f_1$.
Consider a level $w \ge 1$ and let $t_0$ denote the number of counters at level $w$ that contains value zero. Let $r$ denote the total number of counters at level $w$. 
Under the assumption that ideal hash functions are used, multi-level sampling has following properties. 
The expected number of distinct $k$-mers sampled in level $w$ is given by $N_w = F_0/2^w$. 
Furthermore, the expected number of counters with value zero at level $w$ is given by $r  (1 - 1/r)^{N_w}$. Using $t_0$ as an estimate for this, estimate $\hat{F}_0$ can be computed as 
$$
\hat{F}_0 = 2^{w} \frac{ \ln\left(t_0/r\right)}{\ln\left(1 - \frac{1}{r}\right)}
$$
Finally a `good' level $w$ for estimating $F_0$ is chosen with the property that value of $t_0$ at this level is close to $r/2$.
It was shown in \cite{barYosef2002} that this approach yields $\hat{F}_0$ with the property that  $(1-\epsilon) F_0 \le \hat{F}_0 \le (1+ \epsilon) F_0$ with probability at least $ 1 - \delta$. The algorithm has $O(1)$ update time and uses $O(\frac{1}{\epsilon^2} \log(1/\delta) \log(F_0))$ memory.

\ignore{

	Let $t_1$, be the  number of counters in level $w$ each with value 1. The expected number of sampled $k$-mers in level $w$ with multiplicity $1$ in the input is given by $f_1/2^w$. The expected value of $t_1$ is given by
	$
	\frac{f_1}{2^w} \left(1- \frac{1}{R}\right)^{\frac{F_0}{2^w} - 1}.
	$
	Using $t_1$, estimate $\hat{f}_1$ can be computed as
	$$
	\hat{f}_1 = t_1  \cdot {2^w} \left(1- \frac{1}{R}\right)^{1 - \frac{F_0}{2^w}}.
	$$

}

Let $t_i$ be the number of counters at level $w$ each with value $i$ and is collision free. That is, no two distinct $k$-mers hash to any of these counters. The  expected value of $t_i$ is given by 
$
\frac{f_i}{2^w} \left(1- \frac{1}{r}\right)^{\frac{F_0}{2^w} - 1}.
$
Knowing the value of $t_i$, we can thus estimate $f_i$ as
$$
\hat{f}_i = t_i  \cdot {2^w} \left(1- \frac{1}{r}\right)^{1 - \frac{F_0}{2^w}}
$$
Determining $t_0$ and $t_1$ are easy as counters holding values either $0$ or $1$ are collision free by definition.
On the other hand, determining $t_i$ for $i \ge 2$ is difficult because several of the counters with value $i$ could have collisions.  Furthermore, for larger values of $i$, there are several collision possibilities that results in a counter value of $i$ and analyzing each of these possibilities to bound the error is extremely complicated. This makes the algorithm in \cite{kmerstream} inadequate to estimate $f_i$ for $i \ge 2$. 
In \cite{kmerstream}, authors prove that if $f_1 \ge F_0/\lambda$, their estimate $\hat{f}_1$ is such that, $(1-\epsilon) f_1 \le \hat{f}_1 \le (1+\epsilon) f_1$ with probability at least $1-\delta$ 
using $O(\frac{\lambda^2}{\epsilon^2} \log(1/\delta) \log(F_0))$ memory and with 
$O(1)$ update time.

In the following we discuss our Kmerlight algorithm for estimating $f_i$ for $i \ge 2$ also in addition to $F_0$ and $f_1$ and with similar theoretical guarantees.
Kmerlight sketch has $t$ instances, where $t$ is a parameter, and each instance has the following structure.
An instance has $M$ arrays $T_1, \ldots, T_M$ where each of these $M$ arrays correspond to $M$ different sampling levels. We choose $M=64$ in our implementation, which is adequate for  counting up to $2^{64}$ distinct $k$-mers.
Each of these $M$ arrays have $r$ counters, where $r$ is the second parameter for the algorithm.  Each counter of the array is a tuple of the form $\langle v, p\rangle$ where $v\ge 0$ is the counter value and $p$ is a number from $\{0, \ldots, u-1\}$, where $u$ is the third parameter for our algorithm. We use a special counter value `$-1$' to indicate that the counter  is `dirty'.
For counter $T_w[i]$, we use $T_w[i].v$ and $T_w[i].p$ to indicate its $v$ and $p$ values respectively. Initially, all $T_w[i].v$ values are set to zero and all $T_w[i].p$ values are set to `undefined'. 
Sampling level $w \in \{1, \ldots, M\}$, for a $k$-mer is computed using the same approach  as in \cite{barYosef2002, kmerstream}, where a $k$-mer is first hashed using a pairwise independent hash function $h()$. The number of trailing bits in the binary representation of the hashed value starting from the least significant `1' bit is used as its sampling level. For the special case of hashed value being $0$, the $k$-mer is assigned level $M$.
Each instance uses an independent hash functions. 
Let $z$ be the hashed value of a $k$-mer and $w$ be the sampling level assigned to it. Let $x = z/2^w$. 
We map the $k$-mer to a pair of values $(c, j) \in \{0, \ldots, r-1\} \times \{0, \ldots, u-1\}$  given by
$$
c =  \lfloor x/u \rfloor \mod r \mbox{~~~~~~~~~~~and~~~~~~~} j = x \mod u.
$$
where $c$ is the index of the counter at level $w$ and $j$ is the auxiliary information. If the counter $T_w[c]$ is not `dirty' (i.e., $T_w[c].v \ge 0$) then it is updated in the following manner: counter value $T_w[c].v$ is incremented by one if either $T_w[c].p$ equals $j$ or $T_w[c].v$  equals $0$.  In the later case, $T_w[c].p$ is additionally initialized to $j$. If neither of these two conditions hold then the counter $T_w[c]$ is marked `dirty' by assigning $T_w[c].v=-1$. As a consequence, `dirty' counters are discarded from future updates.

Note that every occurrence of the same $k$-mer is mapped to the same $(c, j)$ pair. A counter is marked `dirty' when a collision is detected. The $T_w[c].p$ values help in detecting collisions. Clearly not all collisions can be detected because two different $k$-mers could map to the same $(c, j)$ pair at the same level. A false positive counter is a counter that remains non dirty in spite of collisions.  
We analytically show that the level $w^*$ chosen by the algorithm for estimation is a `good' level such that the set of non dirty counters at level $w^*$ each holding value $i$ have only few false positives. Hence its cardinality is a good estimate for $t_i$. Moreover, the $t_i$ value at this level can be used to estimate $f_i$ with high accuracy. In order to reduce the error probability further, we maintain $t$ independent instances of the above sketch, where $t$ is a parameter, such that all $t$ instances are updated for every input $k$-mer. The $f_i$ estimates are computed from each of these instances and their median value is used as the final estimate $\hat{f}_i$. 

Details of the update and estimation methods are given in Algorithm \ref{algorithm1}. Error bounds achieved by our algorithm are given in Theorem \ref{thm:1}. The proof is given in the supplementary material (Section D) and it assumes truly random hash functions. 
Our algorithm uses logarithmic space, logarithmic time for query and $O(1)$ time per update (for fixed $\lambda$ and $\delta$).
We note that the time and space complexities given in Theorem \ref{thm:1} is for the simultaneous estimation of all $f_i$s with $f_i \ge F_0/\lambda$. We refer the reader to the supplementary material for more discussions on the time and space complexities.

\begin{algorithm} 
\begin{algorithmic}[1]
\STATE \textbf{function} Update($k$-mer a)
\STATE Do the following for each of the $t$ instances of the sketch:
\STATE $z \leftarrow h(a)$
\STATE $w \leftarrow$ $1 + $ Number of trailing zeroes in $z$
\STATE $x \leftarrow \frac{z}{2^{w}} $
\STATE $c \leftarrow \lfloor x/u \rfloor \mod r$
\STATE $j \leftarrow x \mod u$
\STATE Let $T_w[c]$ be $\langle v, p \rangle$
\STATE if $(v \ge 0)$ then 
\STATE ~~~if ($v = 0$) then $T_w[c] = \langle 1, j \rangle$
\STATE ~~~else if $(p \not= j)$ then $T_w[c] = \langle -1, p \rangle$
\STATE ~~~else $T_w[c] = \langle v+1, p \rangle$
\vspace*{0.3cm}

\STATE \textbf{function} Estimate $\hat{F}_0$
\STATE For each instance $l \in \{1, \ldots, t\}$, compute estimate $\hat{F}_0^{(l)}$ as follows:
\STATE $w^* \leftarrow \arg\min_w ||\{c: T_w[c].v = 0\}| - \frac{1}{2}|$
\STATE $p_0 \leftarrow \frac{|\{c: T_{w^*}[c].v = 0\}|}{r}$
\STATE $\hat{F}_0^{(l)} \leftarrow 2^{w^*} \left(\frac{\ln(p_0)}{\ln(1 - 1/r)} \right)$
\STATE $\hat{F}_0 \leftarrow$ median of $\hat{F}_0^{(1)} \ldots \hat{F}_0^{(t)}$.
\RETURN $\hat{F}_0$ 
\vspace*{0.3cm}

\STATE \textbf{function} Estimate $\hat{f}_i$ 
\STATE For each instance $l \in \{1, \ldots, t\}$, compute estimate $\hat{f}_i^{(l)}$ as follows:
\STATE $w^* \leftarrow \arg\max_w |\{c: T_w[c].v = i\}|$
\STATE $p_0 \leftarrow \left(1 - \frac{1}{r} \right)^{\hat{F}_0/2^{w^*}}$
\STATE $p_i \leftarrow \frac{|\{j: T_{w^*}[j].v = c\}|}{r}$
\STATE $\hat{f}_i^{(l)} \leftarrow 2^{w^*} \left( \frac{(r-1) p_i}{p_0} \right)$
\STATE $\hat{f}_i \leftarrow$ median of $\hat{f}_i^{(1)} \ldots \hat{f}_i^{(t)}$.
\RETURN $\hat{f}_i$ 
\end{algorithmic} 
\caption{Computation of $F_0$ and $f_i$ for $i \ge 1$. (Parameters: $t, r, u$)}
\label{algorithm1}
\end{algorithm}

\begin{theorem}
\label{thm:1}
With probability at least $1 - \delta$, estimate $\hat{F}_0$ for $F_0$ and estimates $\hat{f}_i$ for every $f_i$ with $f_i \ge F_0/\lambda$ can be computed 
such that $(1-\epsilon) F_0 \le \hat{F}_0 \le (1+\epsilon) F_0$ and  $(1-\epsilon) f_i \le \hat{f}_i \le (1+\epsilon) f_i$ using
$O(\frac{\lambda}{\epsilon^2} \log(\lambda/\delta) \log(F_0))$ memory and $O(\log(\lambda/\delta))$ time per update.
\end{theorem}

\subsection{Repetitive Regions in the Genome} \label{sec:repeat}
We show applications of $k$-mer abundance histograms computed by Kmerlight. In this section, we apply $k$-mer abundance histogram of reads for de novo estimation of $k$-mer repetitiveness in the underlying genome. 
If segments in the genome are repetitive, then $k$-mers from these segments would have correspondingly scaled multiplicity in the reads under uniform read coverage.
Conversely, large fraction of $k$-mers having higher multiplicity values in the read collection is indicative of sizable repetitive regions in the genome. 
Thus, understanding $k$-mer repetitiveness in the genome for varying values of $k$ can provide useful insights into the repetitive nature of the genome.
Studying with multiple values of $k$ is useful because large values of $k$ would help in excluding short repeats (length less than $k$) from the analysis. 

For a fixed $k$-mer length $k$,  we define the multiplicity (repetitiveness) of any given location $s$ in the genome as the number of times the $k$-mer starting at location $s$ is present across the whole genome.
Let $g_t$ denote the total number of positions each with multiplicity $t$ in the genome and let $g$ denote the genome size. Clearly, $g = \sum_t g_t$.  
Thus, $g_t/g$ is the fraction of the genome locations each with repetitiveness $t$. In other words, $g_t/g$ is the fraction of the genome locations such that the $k$-sized fragment occurring at any of these locations occur in total $t$ times in the genome.
In the following, we propose a method to estimate $g_t$ values from reads. For this, we propose a  simple probabilistic model for $k$-mer abundance histogram of reads. The model parameters include $g_t$ values, which can thus be inferred from observed histogram.

\vspace*{0.2cm}
\noindent
\textbf{Generative Model for $k$-mer Abundance Histogram} 

We propose a simple probabilistic model for the $k$-mer abundance histogram of sequence reads.
For this we follow the generative model proposed in \cite{kmerstream} where the model assumes uniform coverage of the genome and that the $k$-mers are generated from each position of the genome are Poisson distributed as $Poi(\lambda)$. Furthermore, the true $k$-mers are generated at each position as $Poi(\lambda')$ and the erroneous $k$-mers are generated as $Poi(\lambda - \lambda')$. 

With these assumptions, we give a simple model for the abundance histogram of true $k$-mers.
Parameter $\lambda$ is related to total number of $k$-mers $N$ from the reads as
$ \lambda = N/g,$ where $g$ is the genome size. The $k$-mer error rate is given by $\lambda - \lambda'$.
We use the notations $c$ for the read coverage, $l$ for read length, $n$ for total number of reads, $N$ for total number of $k$-mers and $g$ for the genome size. Clearly 
$
N = n(l-k+1).
$
Since $c = nl/g$, we also obtain that 
\begin{eqnarray}\label{eq:lambda}
\lambda = N/g = c (l-k+1)/l
\end{eqnarray}
and
\begin{eqnarray}\label{eq:g}
g = \frac{N l }{c (l-k+1)}
\end{eqnarray}

Let $G_m$ for $m \ge 1$ denote the set of all distinct $k$-mers each occurring with multiplicity $m$ in the underlying genome. Recalling that $g_m$ denote the total number of positions in the genome each with multiplicity $m$, it follows that $g_m = m\cdot |G_m|$. 
Consider a $k$-mer $x$ belonging to $G_m$.
Let $X_1, X_2, \ldots, X_m$ be $m$ random variables each distributed as $Poi(\lambda')$ and denoting number of times $x$ was sampled from each of its $m$ locations in the genome. 
Let $X=X_1 + \cdots + X_m$ denote the total number of occurrences of $x$ in the final collection of true $k$-mers.  By linearity of Poisson, it follows that $X$ is Poisson distributed as $Poi(m\lambda')$. 
That is, each true $k$-mer with multiplicity $m$ in the genome is Poisson distributed as $Poi(m \lambda')$ in the reads.
We use the known fact that $Poi(\lambda)$ has  peak probability at value $\lambda$ with corresponding probability value $\lambda^\lambda e^{-\lambda}/\lambda! \approx 1/\sqrt{2\pi\lambda}$. 
We consider the abundance histogram of all $k$-mers from $G_m$ present in the reads.
It follows that this abundance histogram is expected to have peak at $m \lambda'$ with peak value $|G_m|/\sqrt{2\pi m \lambda'} = g_m / (m \sqrt{2 \pi m \lambda})$.
Consequently, corresponding to  the $k$-mers in the genome with multiplicity $m$, for $m = 1, 2, \ldots$, 
there would be a peak in the abundance histogram
at $m\lambda'$ with peak value $g_m /(m \sqrt{2\pi m \lambda'})$.
In a diploid case, the first peak (at $\lambda'$) would correspond to the heterozygous $k$-mers in the genome.

From this model, $g_m$ values can be easily computed using the corresponding peak positions and peak values in the histogram.
The fractions $g_m/g$ can be computed using estimate for genome size $g$.
Estimate for $g$ can also be obtained from the  $k$-mer abundance histogram, which we discuss in section \ref{sec:kmererror}. 
By using Kmerlight for computing abundance histograms, estimates of $g_m$ for various $k$-mer sizes can be computed efficiently.

%

\vspace*{0.2cm}
\noindent
\textbf{Erroneous $k$-mers}

The generation of erroneous $k$-mers is modelled in \cite{kmerstream} as following $Poi(\lambda-\lambda')$.  It is further assumed in \cite{kmerstream} that erroneous $k$-mers result from single position errors. In other words, an erroneous $k$-mer at a position could be any one of the $3k$ possible candidates.
Consequently,
most erroneous $k$-mers occur only once in the reads and the cardinality of those with multiplicity two or more are significantly lesser, as typically observed in practice.
Hence erroneous $k$-mers contribute to an initial sharp peak at position $1$ in the abundance histogram, followed by peaks due to true $k$-mers with varying multiplicities. This is illustrated in Figure  \ref{fig_spectrum}. Typically, the fraction of  genome with a given $k$-mer multiplicity decreases with increasing multiplicity value and hence the abundance histogram will typically have peaks with decreasing peak values with increasing multiplicities. 
Since $\lambda$ is related to $c$ and $l$ as in eq (\ref{eq:lambda}),  increasing the coverage $c$ results in increased value of  $\lambda$ and of $\lambda'$,  resulting in a wider gap between the peak due to erroneous $k$-mers and remaining peaks. 

\begin{figure}[htbp] 
\begin{center}
\scalebox{0.3}{
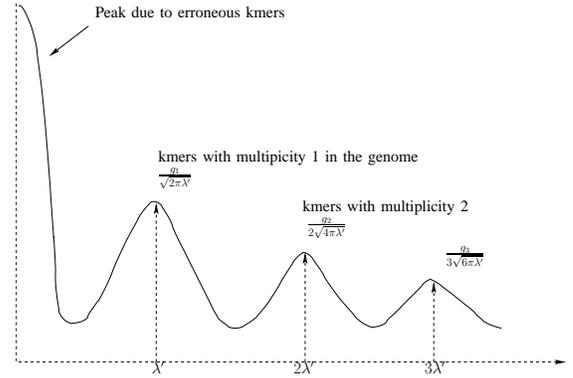}
\caption{Typical $k$-mer abundance histogram of reads.  $x$-axis corresponds to multiplicity values and $y$-axis corresponds to $k$-mer counts. Initial sharp peak is due to read errors and remaining peaks are due to true $k$-mers in the genome with different multiplicities.}
\label{fig_spectrum}
\end{center}
\end{figure}


\subsection{$k$-mer Error Rate Estimation}  \label{sec:kmererror}

$k$-mer error rate is the ratio of total number of erroneous $k$-mers in the sequence reads to the genome size. Understanding $k$-mer error rate gives insights into sequencing 
errors. Fast estimation of $k$-mer error rate from reads was proposed in \cite{kmerstream} by using estimates for $F_0$ and $f_1$ values of $k$-mers in the reads.
For this purpose, a model for erroneous $k$-mer generation was proposed in \cite{kmerstream} and $k$-mer error rate is inferred from this model by numerically solving a set of non linear equations involving estimates for $F_0$ and $f_1$.
Alternatively, 
$k$-mer error rate can be estimated in a straightforward manner if the abundance histogram is available, which can be efficiently computed using Kmerlight.
Let $N'$ and $N_e$ denote the total number of true $k$-mers and total number of erroneous $k$-mers respectively in the sequence reads. Total number of $k$-mers $N$ is thus given by $N=N' + N_e$ and the true $k$-mer rate $\lambda'$ and $k$-mer error rate $\lambda_e$ are given by $\lambda' = N'/g$ and $\lambda_e = N_e/g$ respectively. 
It follows that $\lambda_e$ is given by 
\begin{equation}
\label{eq:ne}
\lambda_e = \frac{\lambda' N_e}{(N-N_e)}
\end{equation}

We recall from  the  model for $k$-mer abundance histogram proposed in the previous section that $\lambda'$ is easily obtained from the peak positions corresponding to true $k$-mers in the histogram. 
Value of $N_e$ can be estimated from the initial peak in the histogram due to erroneous $k$-mers as follows. Estimate $\hat{N}_e$ for $N_e$ is given by $\hat{N}_e = \sum_{i=1}^t i \cdot f_i$. Typically the $f_i$ values for erroneous $k$-mers exhibit sharp decline as $i$ increases. Hence good estimates for $N_e$ can be obtained by considering only first few initial histogram values. 

We remark that this approach avoids detailed modeling of erroneous $k$-mer generation process and solving complex non linear equations as in \cite{kmerstream}.
From $\lambda'$ and $\lambda_e$, we obtain $\lambda = \lambda' + \lambda_e$. Knowing $\lambda$, we can estimate coverage $c$ using eq (\ref{eq:lambda}). Estimate for $c$ can subsequently be used for estimation of $g$ using eq (\ref{eq:g}).


\section{Results}

\subsection {Run time}
We compared run time performance of Kmerlight with the KMC2 $k$-mer counting tool \cite{KMC2015}.  
Kmerlight tool has multi-threaded C++ implementation.
KMC2 was chosen for comparison because KMC2 was shown to outperform other state of the art $k$-mer counting tools \cite{KMC2015} 
in terms of time and space requirements. 
We remark that KMC2 tool computes frequency counts for individual $k$-mers as against the $k$-mer abundance histogram. 
Abundance histogram can nevertheless be computed from individual frequency counts.
Number of threads were kept seven for both tools. KMC2 we considered only the time taken for $k$-mer counting and excluded the additional time for computing the abundance
 histogram.  The run times are given in Table \ref{tab:run1}. 
KMC2 was run with strict memory mode. We give RAM and HDD usage for KMC2. For Kmerlight, only RAM usage is shown since it does not use any HDD space.

The first input contained 328 million reads with read length 90. 
The tools were run on a dual core desktop machine with 4 GB RAM and HDD where each core is a 3.10 GHz Intel i5 processor.   
The input contained  23 billion $21$-mers and 9.2 billion $63$-mers and the respective run times are given in the first two rows of Table \ref{tab:run1}.
KMC2  failed to process the $63$-mers using 1GB RAM and aborted with error.
There is only a minor variation in the Kmerlight run time for different RAM settings. Kmerlight run times for $63$-mers with different memory settings are given in the Table.
The input FASTQ file size was about 77.55 GB and the run times indicate that  Kmerlight processed the file at speeds in the range of 65 MBps to 79 MBps on a desktop machine. These speeds are comparable to the access speed of modern day 7200 rmp hard disks. Kmerlight scales well to genome scale date.
Kmerlight was used to process reads from whole human genome (GRCh38). The reads were generated using ART tool \cite{art2012}  with read length 100 and coverage 50. 
The generated FASTQ file was 320 GB in size and contained 1.45 billion reads. Kmerlight processed 55 billion $63$-mers present in this in 1hr and 20 minutes.

\begin{table}
\caption{Run time comparison of KMC2 and Kmerlight}
\label{tab:run1}
\begin{tabular}{| c | c | c|}
  \hline			
  Input & KMC2  & Kmerlight \\
\hline
\hline
23 billion &	2350 sec -  & 1395 sec (120 MB RAM)  \\
$21$-mers &	(1GB RAM, 20GB HDD) &    \\\hline
9.2 billion  &	Abort - (1GB RAM) & 995 sec (120 MB RAM)\\ 
$63$-mers &	3150 sec -   & 1200 sec (500 MB RAM) \\
 & (2GB RAM, 42GB HDD)	 & 1291 sec (960 MB RAM) \\\hline
144 billion  &	4.5 hr - &  2 hr 10 min (120 MB RAM)\\
$63$-mers & (32 GB RAM, 264GB HDD) & \\\hline
61 billion  &	Abort - (32GB RAM) &  1 hr 10 min (120 MB RAM)\\
$85$-mers &	2 hr 20 min - & \\
& (64GB RAM, 250GB HDD) & \\\hline
\end{tabular}
\end{table}


Next we compared the performance on metagenome scale data. Reads were generated using ART tool \cite{art2012} from about 188039 reference genomes available from the NIH Human Microbiome Project \cite{NIHMicrob}. Read length and coverage were 100 and 50 respectively. The resulting read file was 930 GB in size and contained about 3.8 billion reads. 
These experiments were conducted on a multi-core server with large memory due to larger memory requirement for KMC2 tool. 
The input contained  144 billion $63$-mers and 61 billion $85$-mers and the respective run times are given in the last two rows of Table \ref{tab:run1}.
KMC2  failed to process the $85$-mers using 32GB RAM and aborted with error.

\subsection{Accuracy}
To measure the accuracy of Kmerlight, we compared Kmerlight output with the exact histogram.
In particular, we compared output for different $k$-mer sizes and for different memory usages by Kmerlight.
Reads were generated from human reference chromosome Y (GRCh38) using ART tool \cite{art2012} with read length $100$ and coverage $50$ respectively. 
Figure \ref{figure:histo_true} shows the accuracy of Kmerlight output for $f_4, \ldots, f_{60}$, after processing the $15$-mers  in the input using 500 MB RAM. 
Detailed error analysis of $F_0$ and $f_1, \ldots, f_3$ are given later.
The left plot shows the Kmerlight histogram along with the exact histogram. The shaded bars correspond to exact values. 
Kmerlight values are mean values from 1000 trials. The right side histogram shows the Kmerlight mean output values along with one standard deviation bars.


\begin{figure}[htbp]
\centering
\begin{minipage}[t] {.48\linewidth}
\includegraphics[height=1.05in]{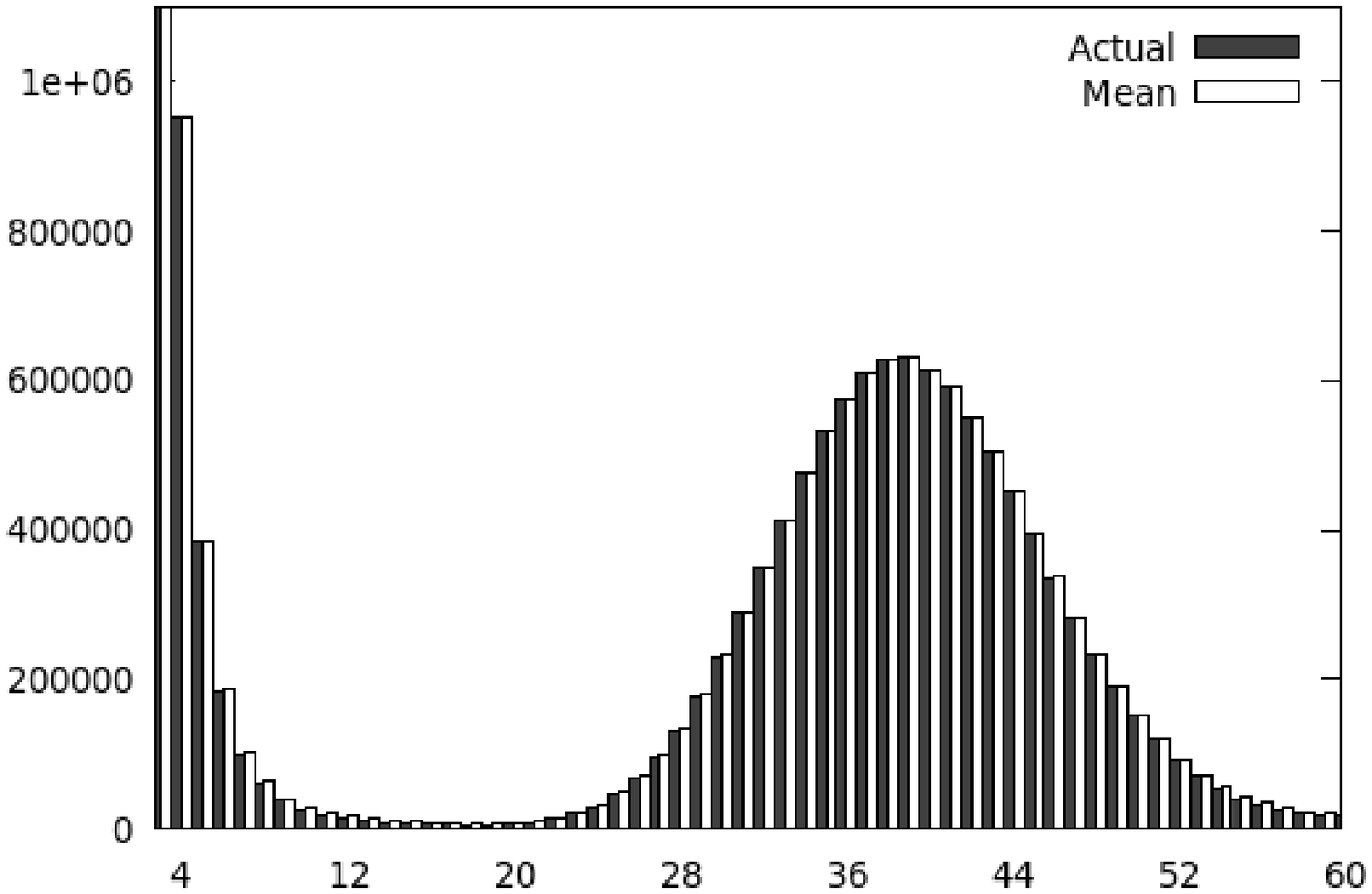}
\end{minipage}
\begin{minipage}[t] {.48\linewidth}
\includegraphics[height=1.05in]{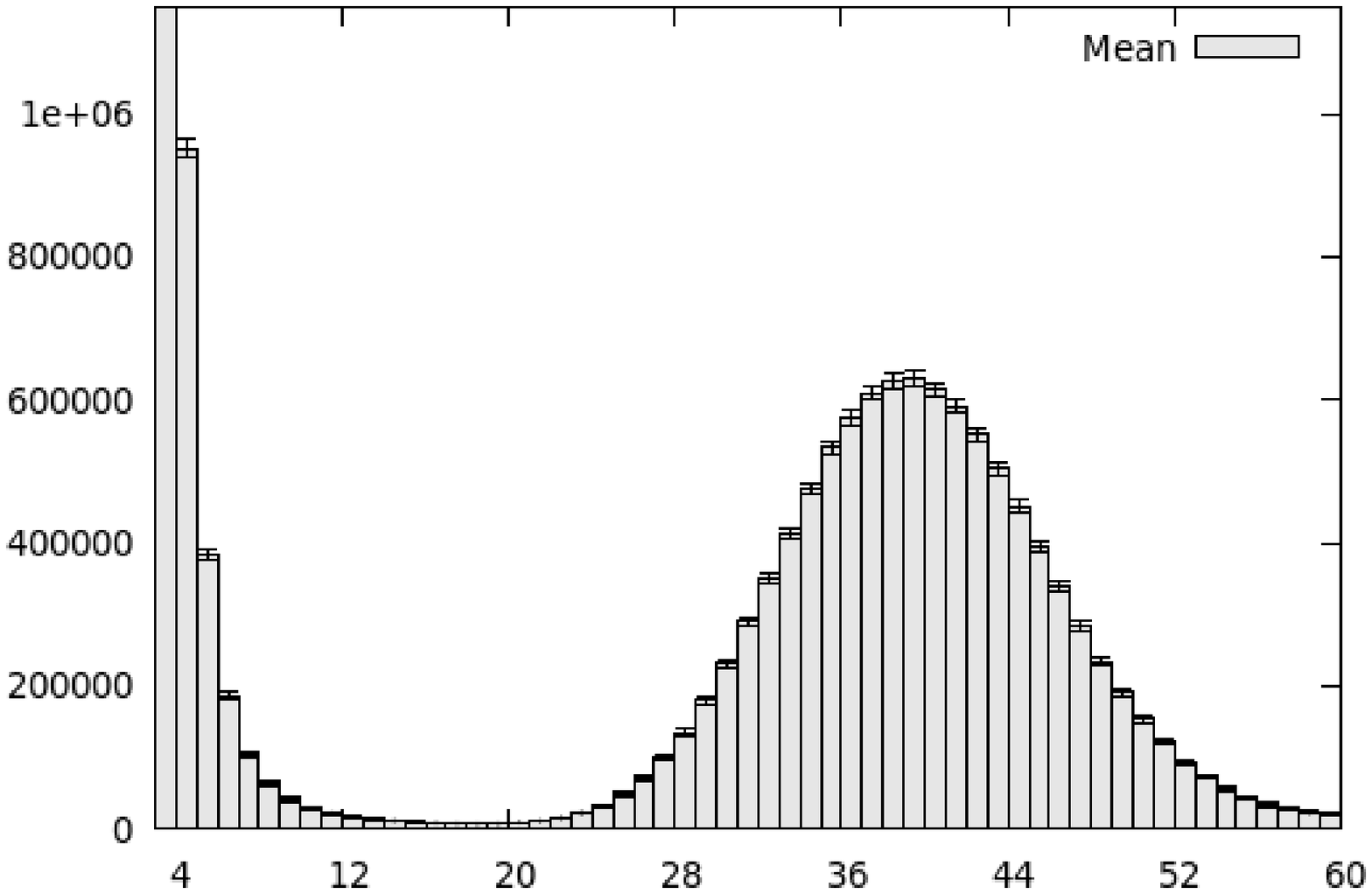}
\end{minipage}
\caption{Histograms of true values and mean output values of Kmerlight. True values are shown side by side as shaded bars.} 
\label{figure:histo_true}
\end{figure}

Figure \ref{figure:hist_rel1} 
shows separate histograms for relative estimate of $F_0, f_1, f_2$ and $f_3$ values of Kmerlight over 1000 trials with respect to true values. Kmerlight used 500 MB RAM setting.
Accuracy plots showing dependence of Kmerlight accuracy other memory settings and on $\lambda$ values are provided in the supplementary material (Section A).

%

\begin{figure}[htbp]
\centering
\begin{minipage}[t]{.48\linewidth}
\centering
\includegraphics[height=1.5in]{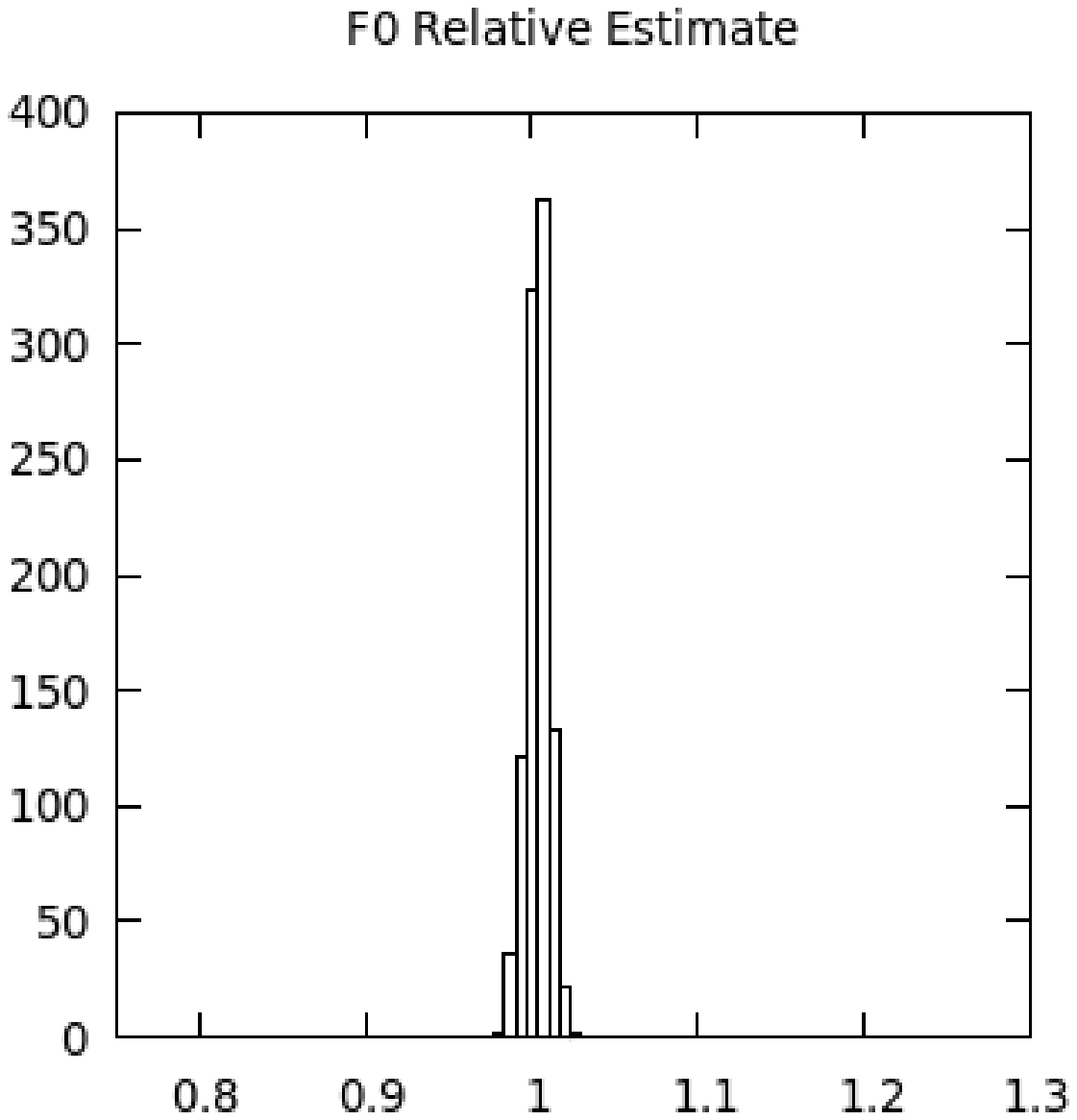}
\end{minipage}
\hfil
\begin{minipage}[t]{.48\linewidth}
\centering
\includegraphics[height=1.5in]{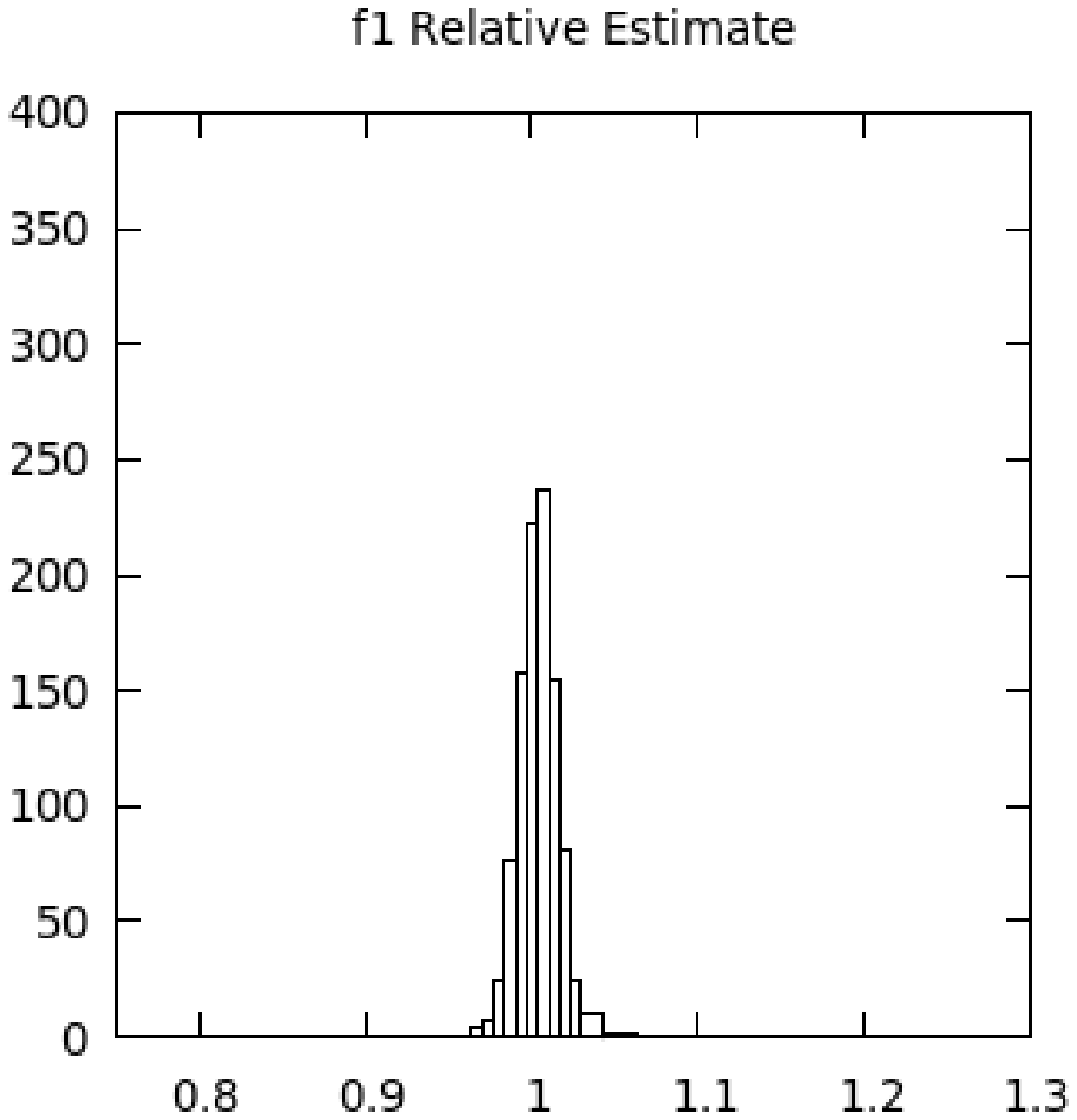}
\end{minipage}
\begin{minipage}[t]{.48\linewidth}
\centering
\includegraphics[height=1.5in]{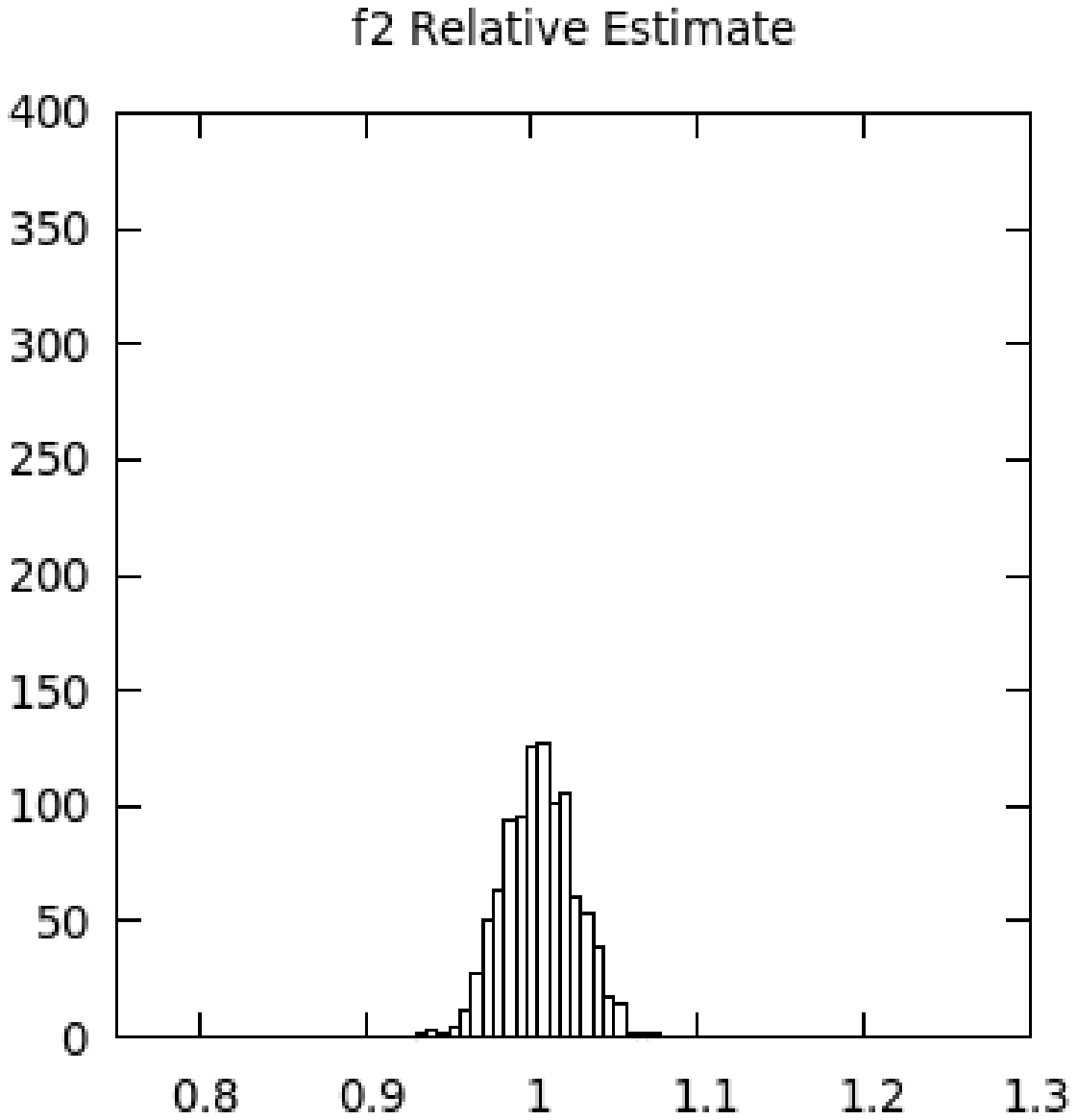}
\end{minipage}
\hfil
\begin{minipage}[t]{.48\linewidth}
\centering
\includegraphics[height=1.5in]{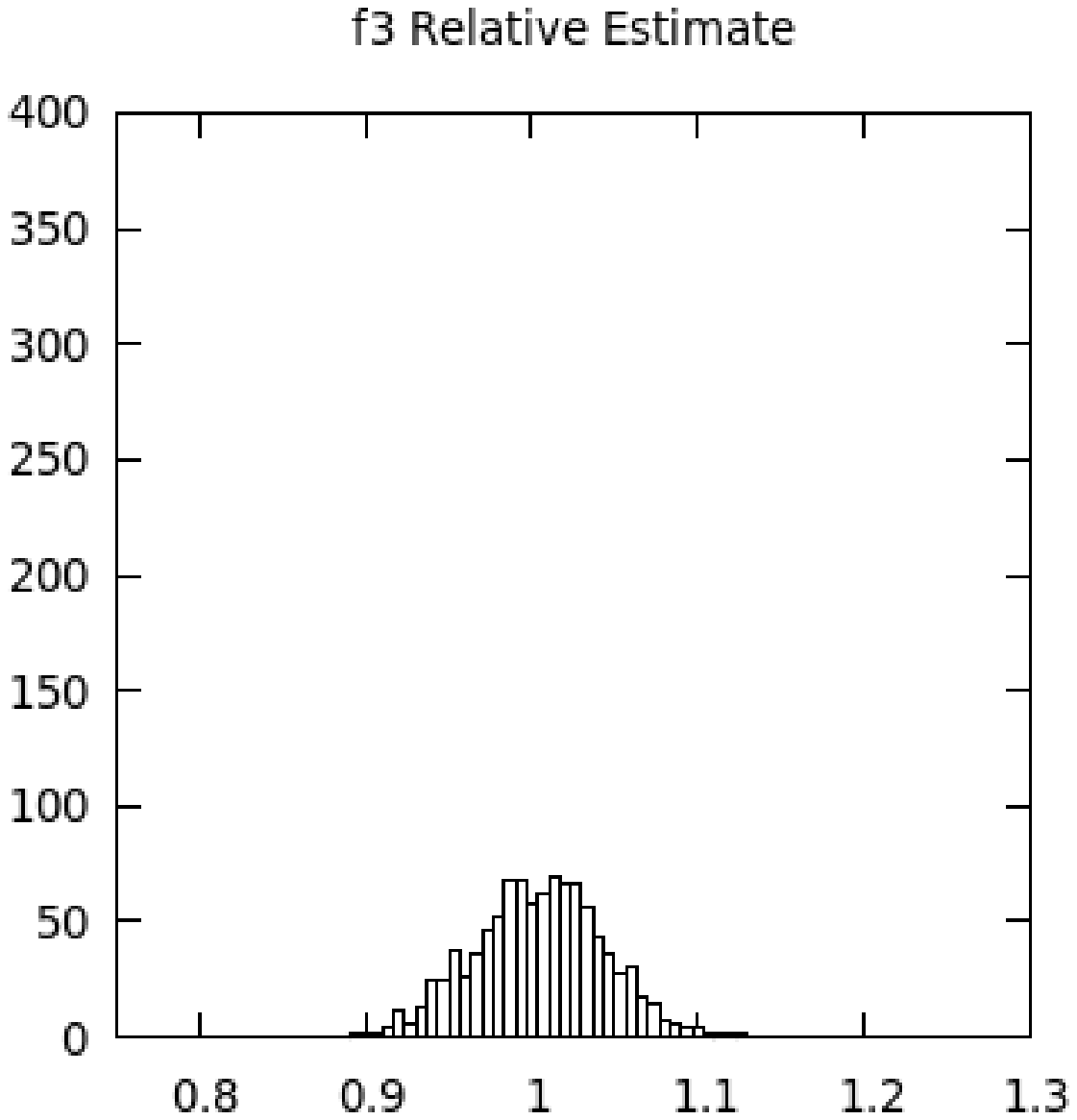}
\end{minipage}
\caption{Histograms showing relative estimate of Kmerlight output for $F_0, f_1, f_2$ and $f_3$ respectively for 1000 trials by Kmerlight with 500 MB RAM setting.}
\label{figure:hist_rel1}
\end{figure}

\ignore{

	\begin{figure*}[t!]
	\centering
	\begin{subfigure}[t]{.48\textwidth}
	\centering
	\includegraphics[height=2in]{f_r16_k21.eps}
	\caption{Kmer spectra }
	\label{figure:spectra}
	\end{subfigure}
	\begin{subfigure}[t]{.48\textwidth}
	\centering
	\includegraphics[height=2in]{f_r18_k21.eps}
	\caption{Kmer spectra }
	\label{figure:spectra}
	\end{subfigure}
	\caption{Mean values for $f_3, \cdots f_6$ averaged over 1000 trials. For each $f_i$, the true value is also shown side by side in red color. The left hand side plot corresponds to  120 MB sketch size and right hand side plot corresponds to 500 MB sketch size. Value of $k$ is 21.}
	\end{figure*}
}

\subsection{Repeats in Chromosome Y}
We used the probabilistic model and the approach discussed in Section \ref{sec:repeat} for de novo estimation of $k$-mer repeats in human reference chromosome Y (GRCh38) from reads.
Validity of the model and the effectiveness of using histogram computed by Kmerlight in place of exact histogram were studied in these experiments.
Reads from chromosome Y were generated using ART tool with read length $100$ and coverage value $50$. 
The built-in profile of Illumina HiSeq 2500 system was used by the ART tool.

Table \ref{tab:gvals} gives the relative errors in estimation of $g_1, g_2$ and $g_3$ values from reads using the model for $k=15$. The estimated values are compared against the true values computed directly from the chromosome Y sequence data. Relative errors are provided separately for estimation using the exact abundance histogram and using Kmerlight computed histograms with different memory settings.
Kmerlight errors are averaged over 1000 trials and the standard deviations are also indicated in the table.
These results indicate that our model is suitable for de novo estimation of $g_i$ values and Kmerlight can be used for the efficient computation of abundance histograms for this purpose.
We refer to the supplementary material (Section B) for additional details including histogram plots, comparison of true peak values and peak values inferred using the model, true $g_i$ values etc.


\begin{table}
\caption{Relative errors in $g_i$ estimation using exact histogram and Kmerlight computed histograms with 500 MB and 940 MB memory settings.
Standard deviations are given inside the brackets.}
\begin{tabular}{| l | c | c | c | }
  \hline			
  $g$ &  Rel. Error  & Rel. Error  & Rel Error \\
   & (exact histogram) &  Kmerlight (500MB) & Kmerlight (940 MB) \\
   &  &   (S.D.) &  (S.D)\\
\hline
\hline
$g_1$ & 0.02  &  0.02 (0.009) & 0.02 (0.011)\\\hline
$g_2$ & 0.02 &  0.042 (0.028)  & 0.02 (0.015) \\\hline
$g_3$ & 0.028 &  0.208 (0.062) & 0.11 (0.043) \\\hline
\end{tabular}
\label{tab:gvals}
\end{table}


\subsection{$k$-mer Error Rate Estimation}
We used Kmerlight to estimate $k$-mer error rate from reads. Reads were generated from human reference chromosome Y (GRCh38) using ART tool with coverage 50 and read length 100. 
The built-in profile of Illumina HiSeq 2500 system was used by the ART tool for read generation.
 The $k$-mer abundance histogram  was then computed using Kmerlight with $k=15$. 
Initial values of this histogram were then used to estimate $N_e$, which denotes the total number of erroneous $k$-mers. In particular, initial 10 histogram values were used for $N_e$ estimation. The first true k-mer peak was observed at $\lambda' = 39$ in the abundance histogram. The $k$-mer error rate $\lambda_e$ was then estimated using eq (\ref{eq:ne}). The average value of $\lambda_e$  over 1000 trials  was obtained as $3.58$. Using exact histogram was used in place of Kmerlight histogram, value of  $\lambda_e$ obtained was 3.52, which is close to the Kmerlight estimate.
 To validate the $k$-mer error rate estimate, we estimated the genome length $g$ using eq (\ref{eq:lambda}) and using the relation $\lambda = \lambda' + \lambda_e$, with $\lambda'=39$ and $\lambda_e=3.58$. 
Alternatively, we used eq (\ref{eq:g}) to estimate $g$ by using the known values $c=50$ and $l=100$ used for read generation. 
The total number of valid nucleotides in the input chromosome Y data was also counted. All the three values for $g$ were closeby and 
the difference between any two of them was within $1\%$. This
indicates high accuracy of the estimated genome length and also the estimated $k$-mer error rate.

\section{Discussions}

\ignore{
	12. Shi H, Schmidt B, Liu W, Muller-Wittig W: A parallel algorithm for error
	correction in high-throughput short-read data on CUDA-enabled
	graphics hardware. J Comput Biol 2010, 17:603-615.
	13. Li R, Zhu H, Ruan J, Qian W, Fang X, Shi Z, Li Y, Li S, Shan G, Kristiansen K,
	Li S, Yang H, Wang J, Wang J: De novo assembly of human genomes
	with massively parallel short read sequencing. Genome Res 2010,
	20:265-272.
	14. Yang X, Dorman K, Aluru S: Reptile: representative tiling for short read
	error correction. Bioinformatics 2010, 26:2526-2533.
}

We proposed the first streaming algorithm Kmerlight for efficient generation of $k$-mer abundance histogram. We provide analytical bounds for the error margins. We also show applications of Kmerlight in the de novo estimation of $k$-mer repeats in the genome using a model that we propose and also in the estimation of $k$-mer error rate and genome length.
We discuss here few additional applications of Kmerlight.

Spectral alignment techniques for read error correction 
in de novo sequencing projects
\cite{ChaissonPevzner2008, chaisson2004, Pevzner2001} 
 usually depend on the set of `trusted' $k$-mers $G^m_k$ in the reads as an approximation to the spectrum $G_k$ which is the set of all $k$-mers in the underlying genome. The set $G^m_k$ denotes the set of all $k$-mers whose frequency is above threshold $m$. The threshold $m$ is usually determined from the region of erroneous $k$-mers in the abundance histogram. Kmerlight allows efficient generation of abundance histograms for various $k$ values.

For choosing appropriate $k$-mer length in de Bruijn based assemblers, Chikhi  et al. \cite{chikhi2013Genie} propose that the most appropriate choice of $k$ is the one that provides maximum number of distinct true $k$-mers to the assembler. This has been shown to yield excellent assembly on a diverse set of genomes. Recalling that $F_0$ denotes the total number of distinct $k$-mers in the sequence reads, $F_0$ can be written as $F_0 = F^e_0 + F'_0$, where $F^e_0$ and $F'_0$ are the total number of distinct erroneous $k$-mers and true $k$-mers respectively. Thus, the objective is to choose $k$ that maximizes $F'_0$. 
In order to do this efficiently, they compute an approximate $k$-mer abundance histogram using down sampling. This is done to overcome the computational overhead of computing exact abundance histograms for different $k$ values. 
This method has been incorporated into the KmerGenie tool. 
Alternatively,  $F'_0$ can be estimated using estimates of $F_0$ and $F^e_0$ as $F'_0 = F_0 - F^e_0$.
Estimate for $F_0^e$ can be obtained by summing values in the abundance histogram region due to erroneous $k$-mers. Kmerlight can be used here for efficient generation of histograms for multiple $k$ values with high accuracy without requiring any down sampling. More details are provided in the supplementary material (Section C).
We remark that estimating $F_0^e$ say from $k$-mer error rate is infeasible because $k$-mer error rate estimates only the total number of erroneous $k$-mers as opposed to the total number of distinct erroneous $k$-mers. 



\bibliographystyle{IEEEtran}
\bibliography{ref}

\clearpage
\newpage
\onecolumn

\centerline{\huge{{Supplementary Material}}}
\vspace*{0.3cm}
\centerline{\Large{{Kmerlight: fast and accurate $k$-mer abundance estimation}}}


\appendices
\section{Additional Error Plots}

We provide additional error plots for the Kmerlight output. 
Input reads were generated from human reference chromosome Y (GRCh38) using ART tool \cite{art2012} with read length $100$ and coverage $50$ respectively. 

Figure \ref{figure:histo_sd} shows the accuracy of the Kmerlight output for $f_4, \ldots, f_{60}$, after processing the $15$-mers  in the input using 120MB RAM setting. 
Detailed error analysis of $F_0$ and $f_1, \ldots, f_3$ are given later.
The left plot shows the Kmerlight histogram along with the exact histogram. The shaded bars correspond to exact values. 
Kmerlight values are mean values from 1000 trials. The right side histogram shows the Kmerlight mean output values along with one standard deviation bars.

\begin{figure}[htbp]
\centering
\begin{minipage}[t] {.48\linewidth}
\includegraphics[height=2in]{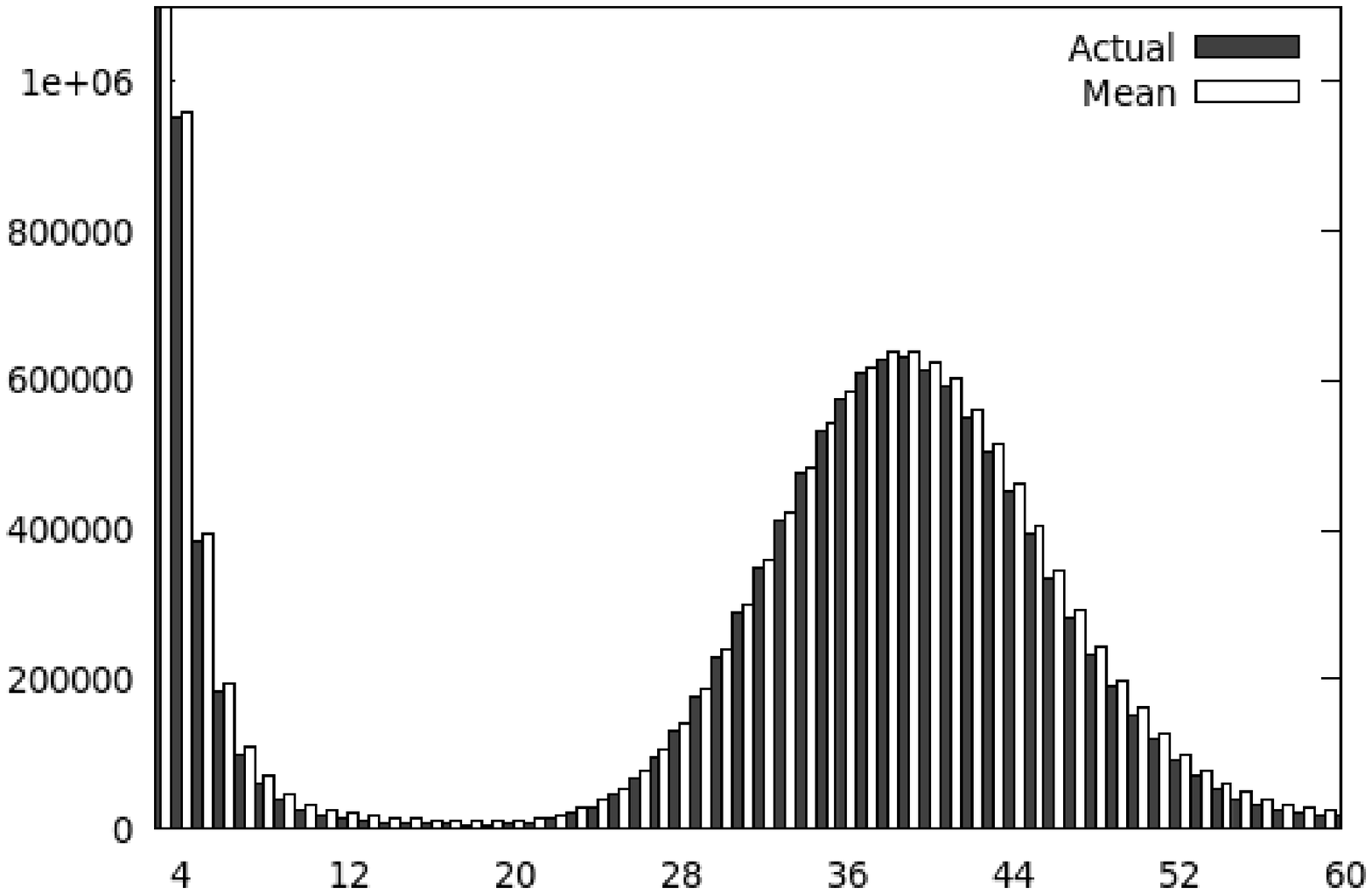}
\end{minipage}
\begin{minipage}[t] {.48\linewidth}
\includegraphics[height=2in]{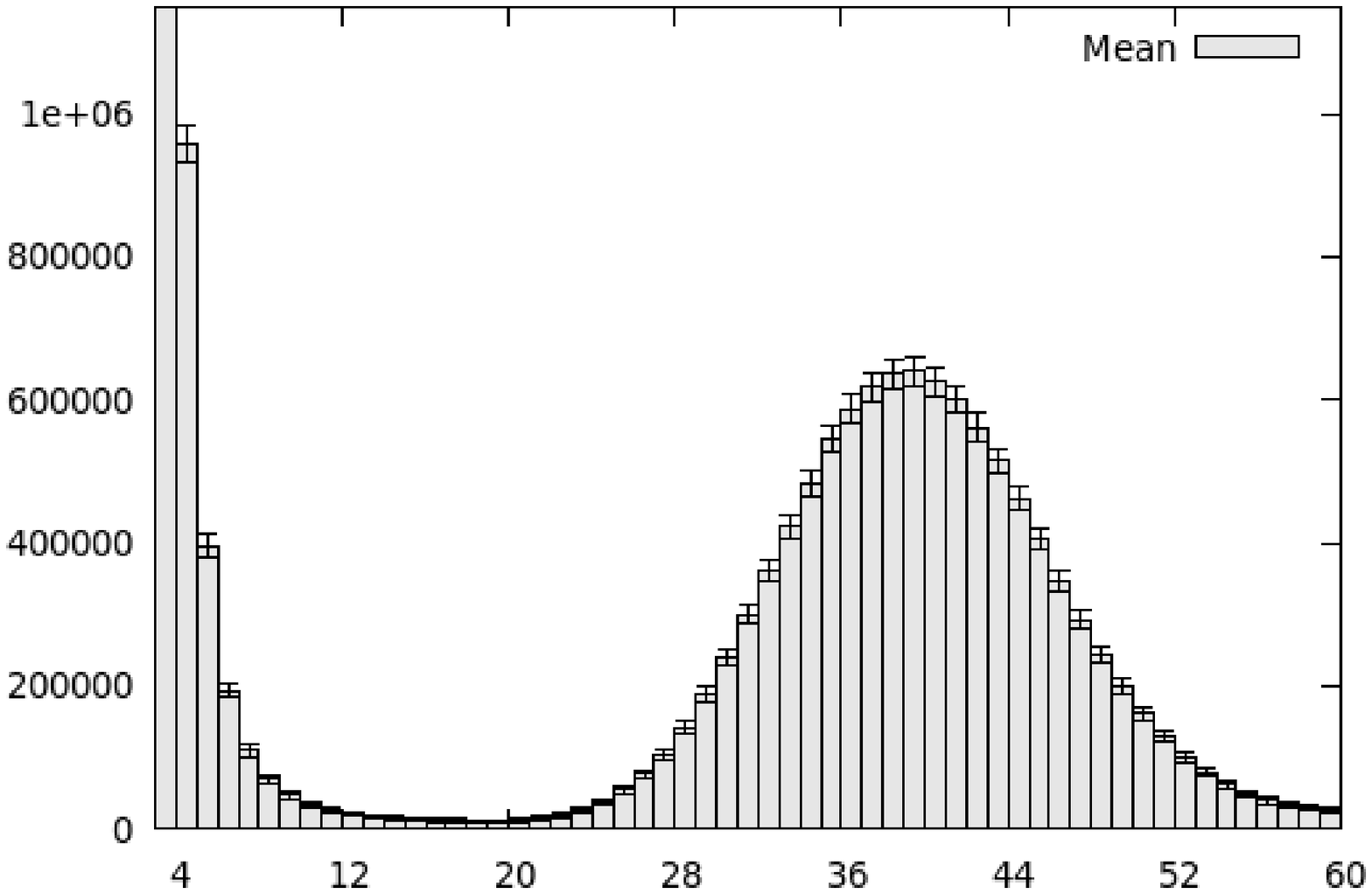}
\end{minipage}
\caption{Histograms of true values and mean output values of Kmerlight using 120MB RAM setting. True values are shown side by side as shaded bars.} 
\label{figure:histo_sd}
\end{figure}

Figure \ref{figure:hist_rel2} 
shows separate histograms for relative estimate of $F_0, f_1, f_2$ and $f_3$ values of Kmerlight over 1000 trials with respect to true values. Kmerlight used 120MB MB RAM setting.

\begin{figure}[htbp]
\centering
\begin{minipage}[t]{.48\linewidth}
\centering
\includegraphics[height=1.8in]{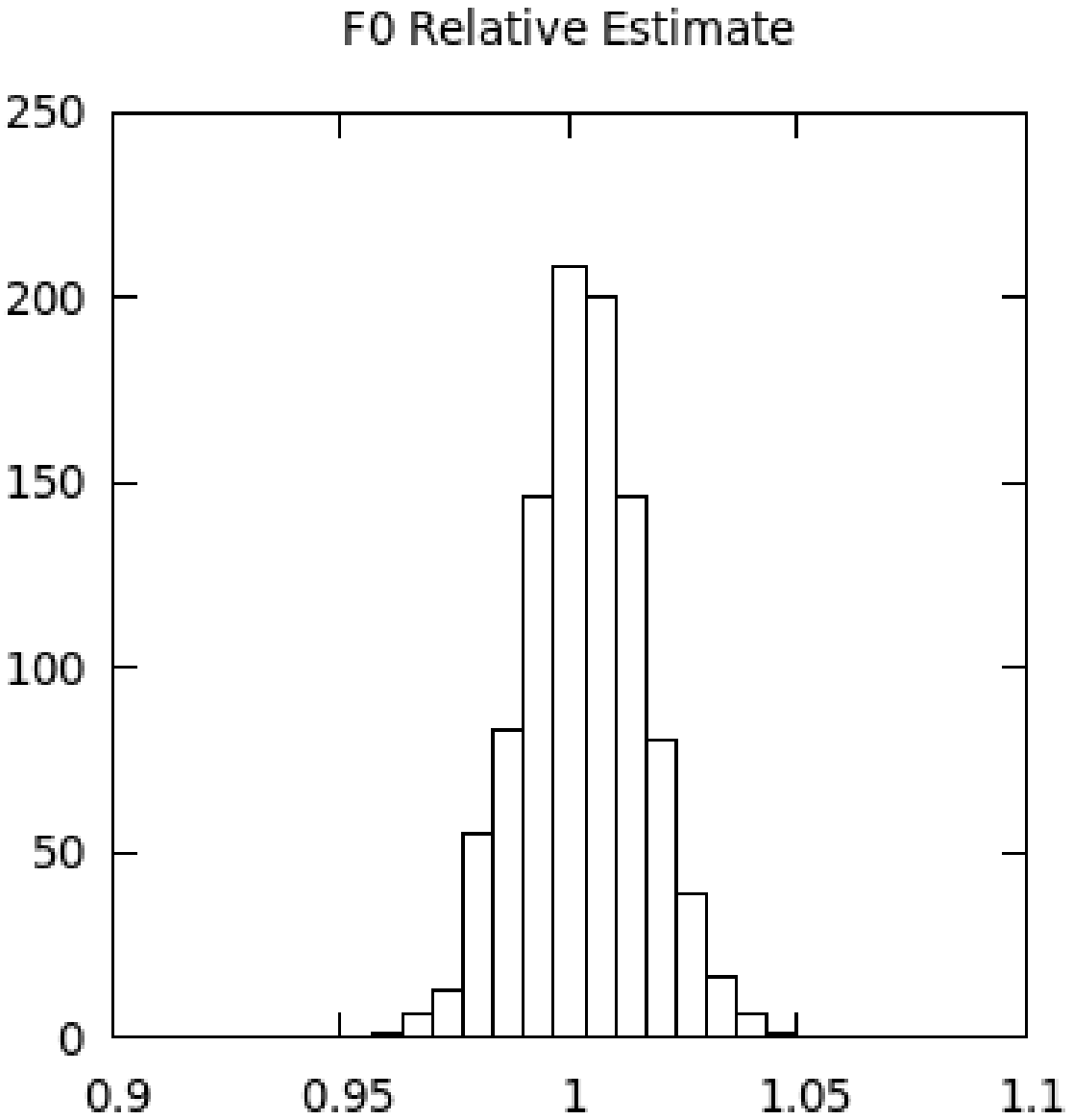}
\end{minipage}
\hfil
\begin{minipage}[t]{.48\linewidth}
\centering
\includegraphics[height=1.8in]{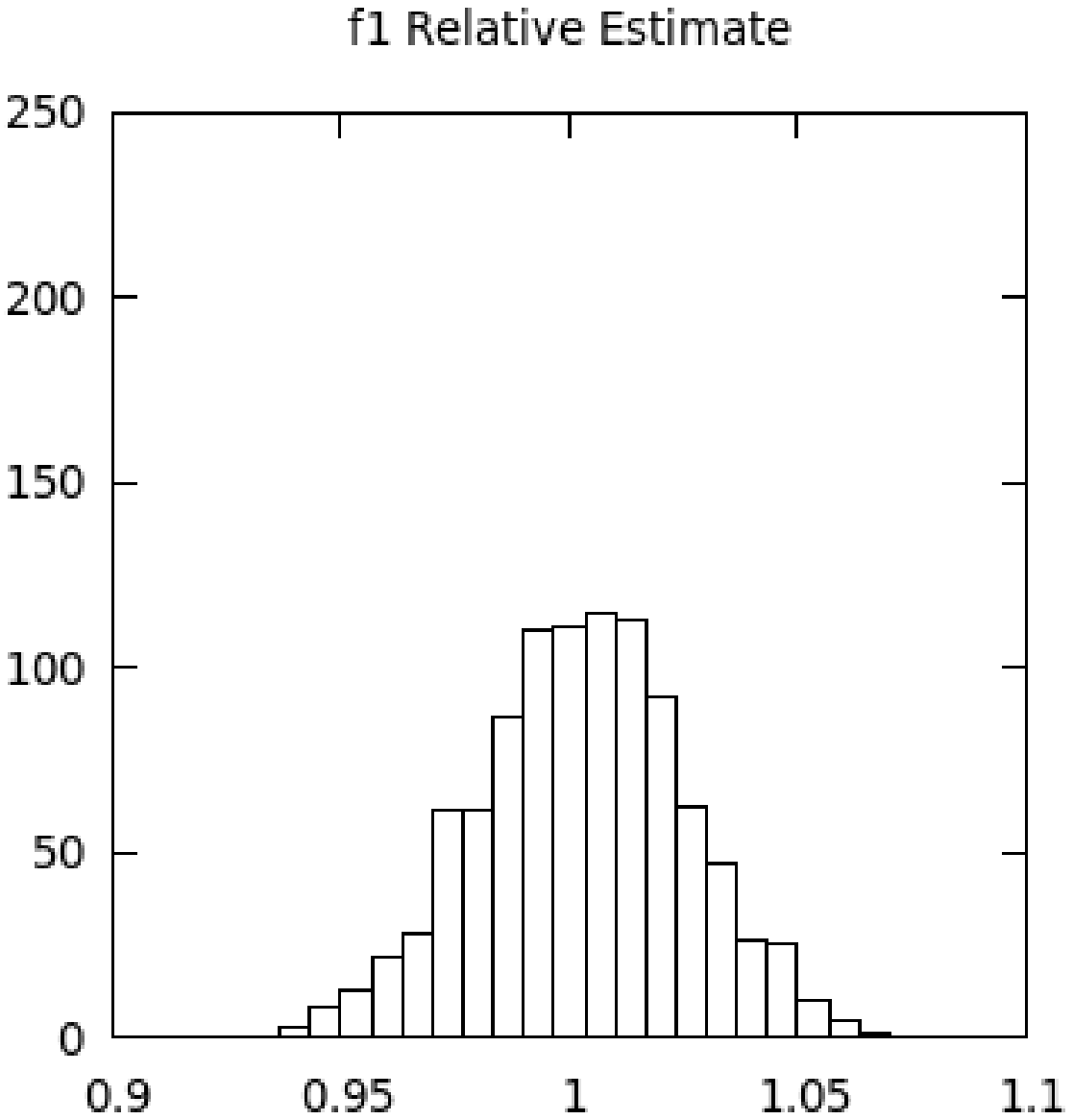}
\end{minipage}
\begin{minipage}[t]{.48\linewidth}
\centering
\includegraphics[height=1.8in]{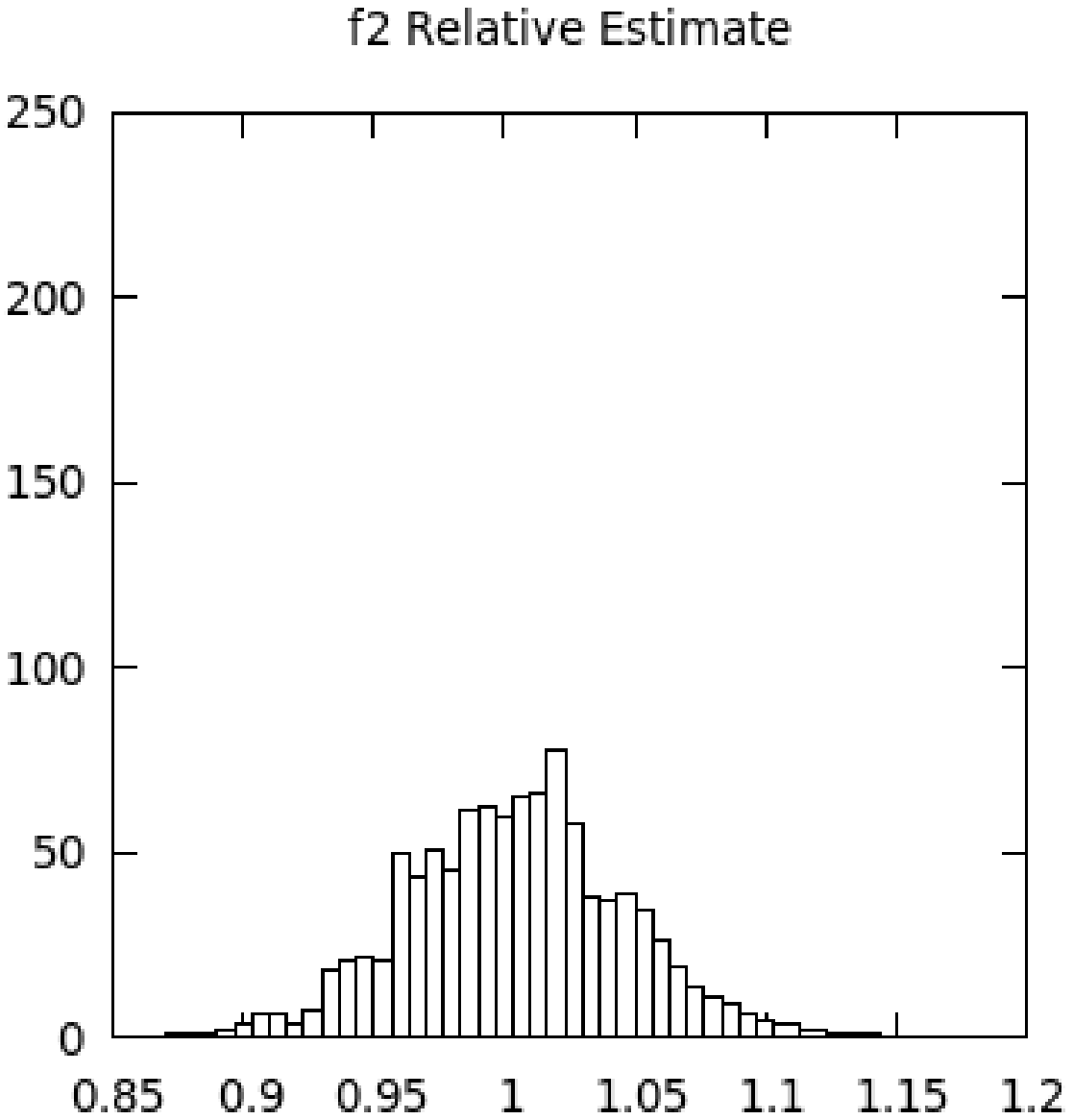}
\end{minipage}
\hfil
\begin{minipage}[t]{.48\linewidth}
\centering
\includegraphics[height=1.8in]{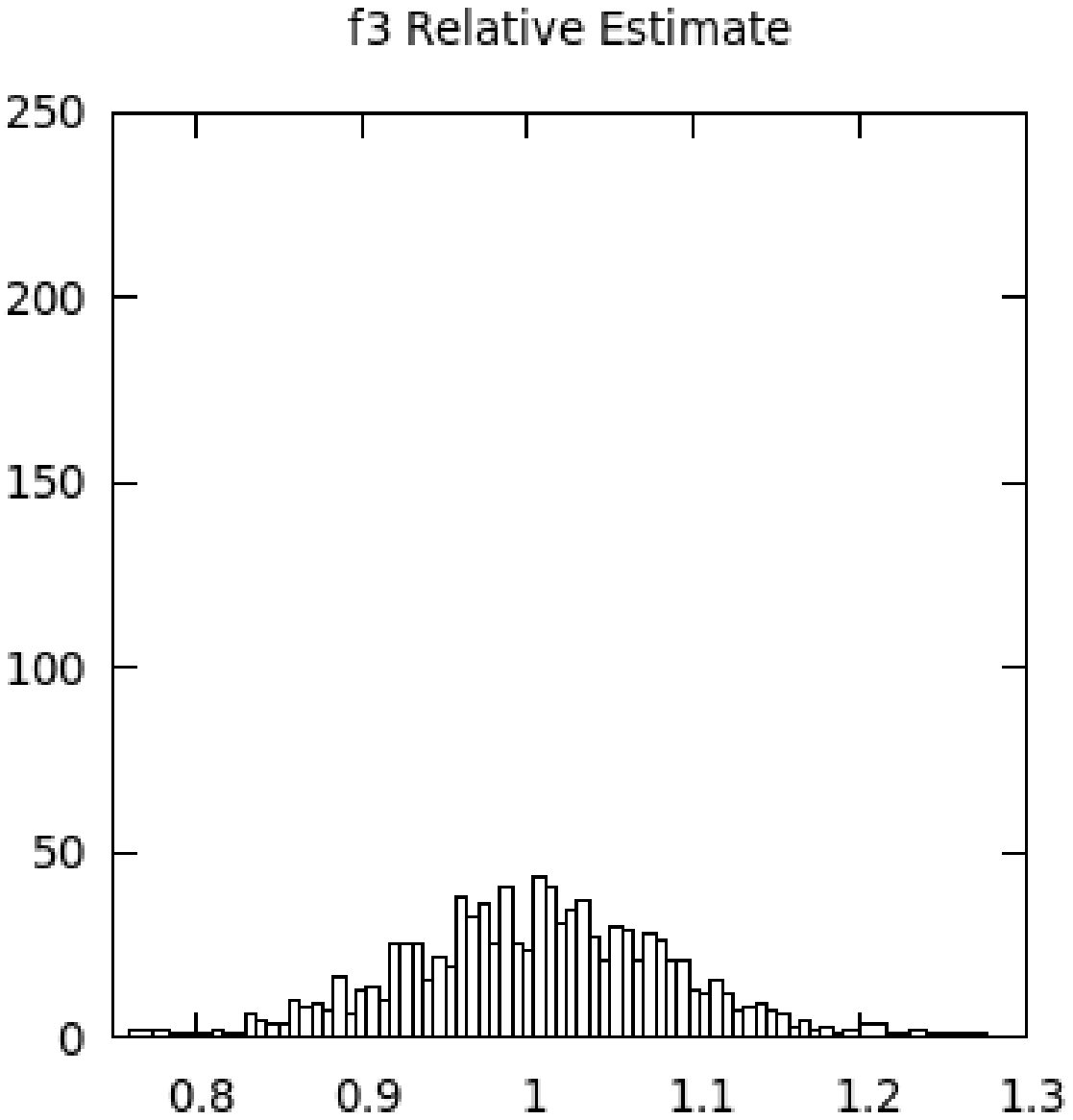}
\end{minipage}
\caption{Histograms showing relative estimate of Kmerlight output for $F_0, f_1, f_2$ and $f_3$ respectively for 1000 trials by Kmerlight with 120 MB RAM setting.}
\label{figure:hist_rel2}
\end{figure}


\vspace*{0.2cm}
\noindent
\emph{Dependence on $\lambda$}
\vspace*{0.2cm}

By Theorem \ref{thm:1}, errors in  $f_i$ estimates decrease with increasing Kmerlight sketch size and increase with increasing $\lambda = F_0/f_i$.  
Estimation accuracy of $f_i$ values were further analyzed with respect to their associated $\lambda = \lceil F_0/f_i \rceil$ ratios. 
Values $f_1, \ldots, f_{1000}$ were considered for this and Figures \ref{figure:lambda} and \ref{figure:lambda-scatter} illustrate the dependence of their estimation accuracy on $\lambda$.
In these figures, the $f_1, \ldots, f_{1000}$ values computed by Kmerlight were segregated with respect to their associated $\lambda = \lceil F_0/f_i \rceil$ ratios. We note that the $\lambda$ associated with an $f_i$ value is obtained by taking the ratio of $F_0$ to the true $f_i$ value. 
Range of $\lambda$ was restricted to $[1, \ldots, 500]$.  

In Figure \ref{figure:lambda}, two plots are provided to show the dependence of estimation error on $\lambda$. The left side plot is for $k=15$ and the right side plot is for $k=21$.
In each plot, the mean relative error as well as its standard deviation for all $f_i$ values associated with a given $\lambda$ value are plotted. For any fixed $\lambda$, estimates from all 1000 trials for each $f_i$ value associated with $\lambda$ were included in the mean and standard deviation calculations.  All $\lambda$ values need not have associated $f_i$ values and hence a piece wise linear plot of the values for $\lambda$s with non empty set of $f_i$ values associated with them is given.
Different Kmerlight memory settings, viz., 120MB ($r=16$) and 500MB ($r=18$), are considered in each of the two plots. As seen in Figure \ref{figure:lambda}, relative errors increase with increasing $\lambda$ and decrease with increasing  memory size. 

\ignore{
	\begin{figure}[htbp]
	\centering
	\includegraphics[height=2.8in]{lambda_k15_r16_r18.eps}
	\caption{Piece wise linear plots for mean and standard deviation of relative errors segregated by $\lambda = F_0/f_i$ values. Solid lines correspond to parameter $r=16$ (120MB RAM) and dotted lines correspond to $r=18$ (500MB RAM) for Kmerlight.}
	\label{figure:lambda}
	\end{figure}

	\begin{figure}[htbp]
	\centering
	\begin{minipage}[t]{.48\linewidth}
	\includegraphics[height=2.4in]{lambda_k15_r16_r18.eps}
	\end{minipage}
	\begin{minipage}[t]{.48\linewidth}
	\centering
	\includegraphics[height=2.4in]{lambda_k21_r16_r18.eps}
	\end{minipage}
	\caption{Piece wise linear plots for mean and standard deviation of relative errors segregated by $\lambda = F_0/f_i$ values. Solid lines correspond to parameter $r=16$ (120MB RAM) and dotted lines correspond to $r=18$ (500MB RAM) for Kmerlight.}
	\label{figure:lambda}
	\end{figure}

}

\begin{figure}[htbp]
\centering
\begin{minipage}[t]{.48\linewidth}
\includegraphics[height=2.4in]{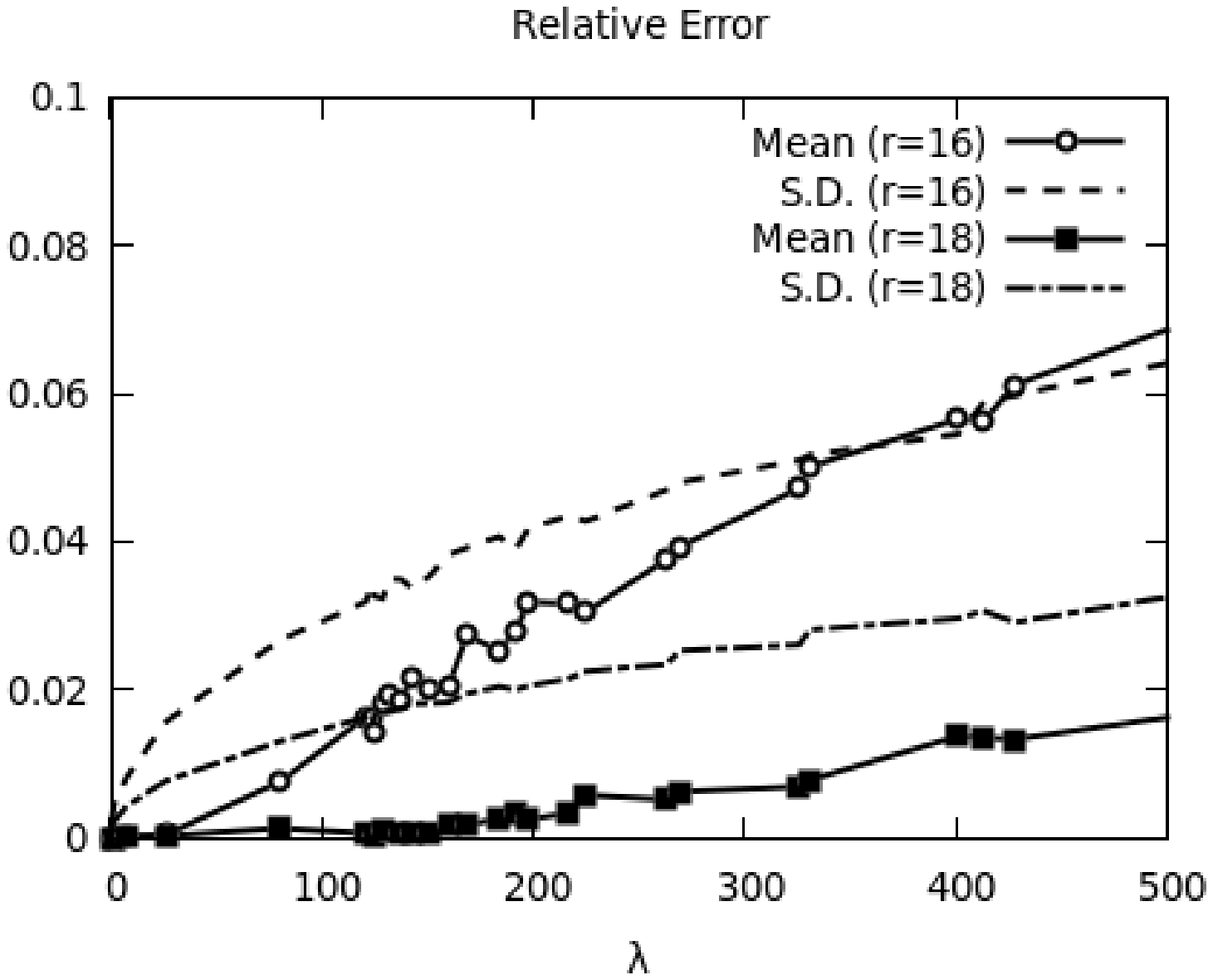}
\end{minipage}
\begin{minipage}[t]{.48\linewidth}
\centering
\includegraphics[height=2.4in]{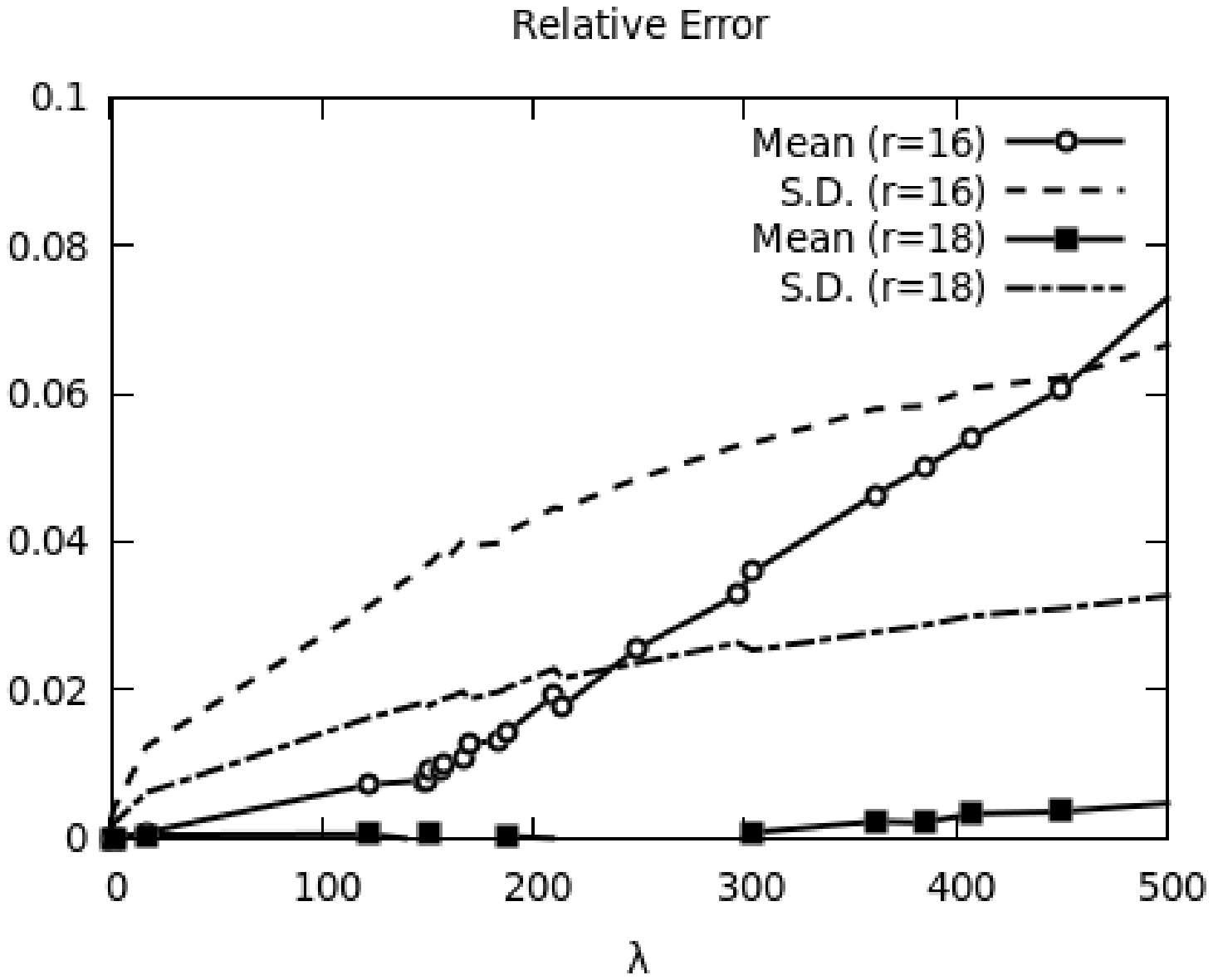}
\end{minipage}
\caption{Piece wise linear plots for mean and standard deviation of relative errors segregated by $\lambda = F_0/f_i$ values. Parameters  $r=16$ and $r=18$ correspond to 120MB and 500MB RAM settings respectively for Kmerlight.}
\label{figure:lambda}
\end{figure}

Figure \ref{figure:lambda-scatter} provides scatter plots of the mean and standard deviation values plotted in Figure \ref{figure:lambda}. The left side and right side scatter plots given in Figure \ref{figure:lambda-scatter} are for memory settings 120MB ($r=16$) and 500MB ($r=18$) respectively of Kmerlight. Both plots contain scatter plots for $k=15$ and $k=21$.  
These plots indicate that already with memory size of 500MB RAM, $f_i$ estimates have high accuracy.

\ignore{

	\begin{figure}[htbp]
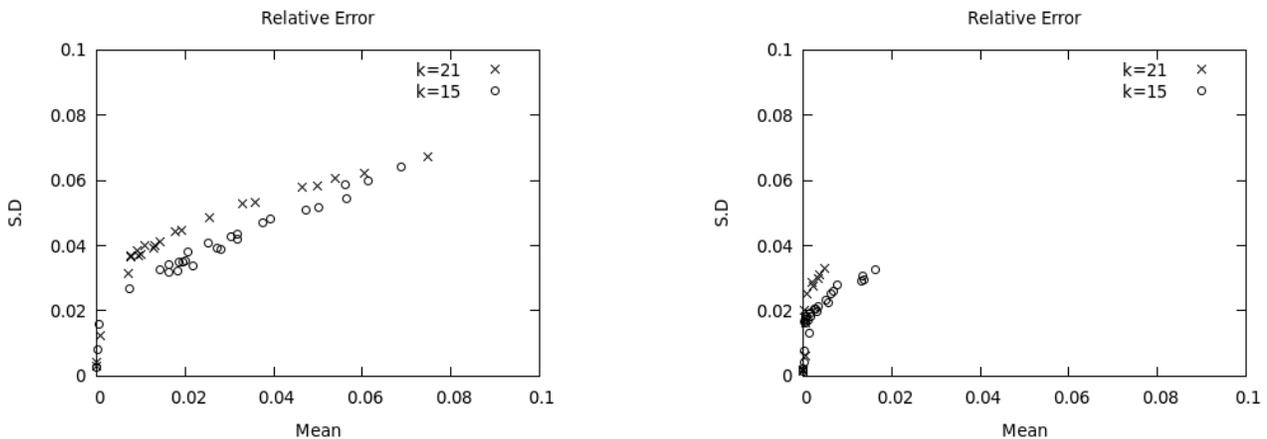

	\centering
	\begin{minipage}[t]{.48\linewidth}
	\includegraphics[height=2.4in]{lambda_points_r16_k15_k21.eps}
	\end{minipage}
	\begin{minipage}[t]{.48\linewidth}
	\centering
	\includegraphics[height=2.4in]{lambda_points_r18_k15_k21.eps}
	\end{minipage}
	\caption{
	Scatter plots for mean and standard deviation of relative errors plotted in Fig \ref{figure:lambda}.  Left and right plots correspond to $r=16$ (120MB) and $r=18$ (500MB)  memory settings respectively for Kmerlight. Plots for both $k=15$ and $k=21$ are provided in each.}
	\label{figure:lambda-scatter}
	\centering
	\end{figure}

}

\begin{figure}[htbp]
\centering
\begin{minipage}[t]{.48\linewidth}
\includegraphics[height=2.4in]{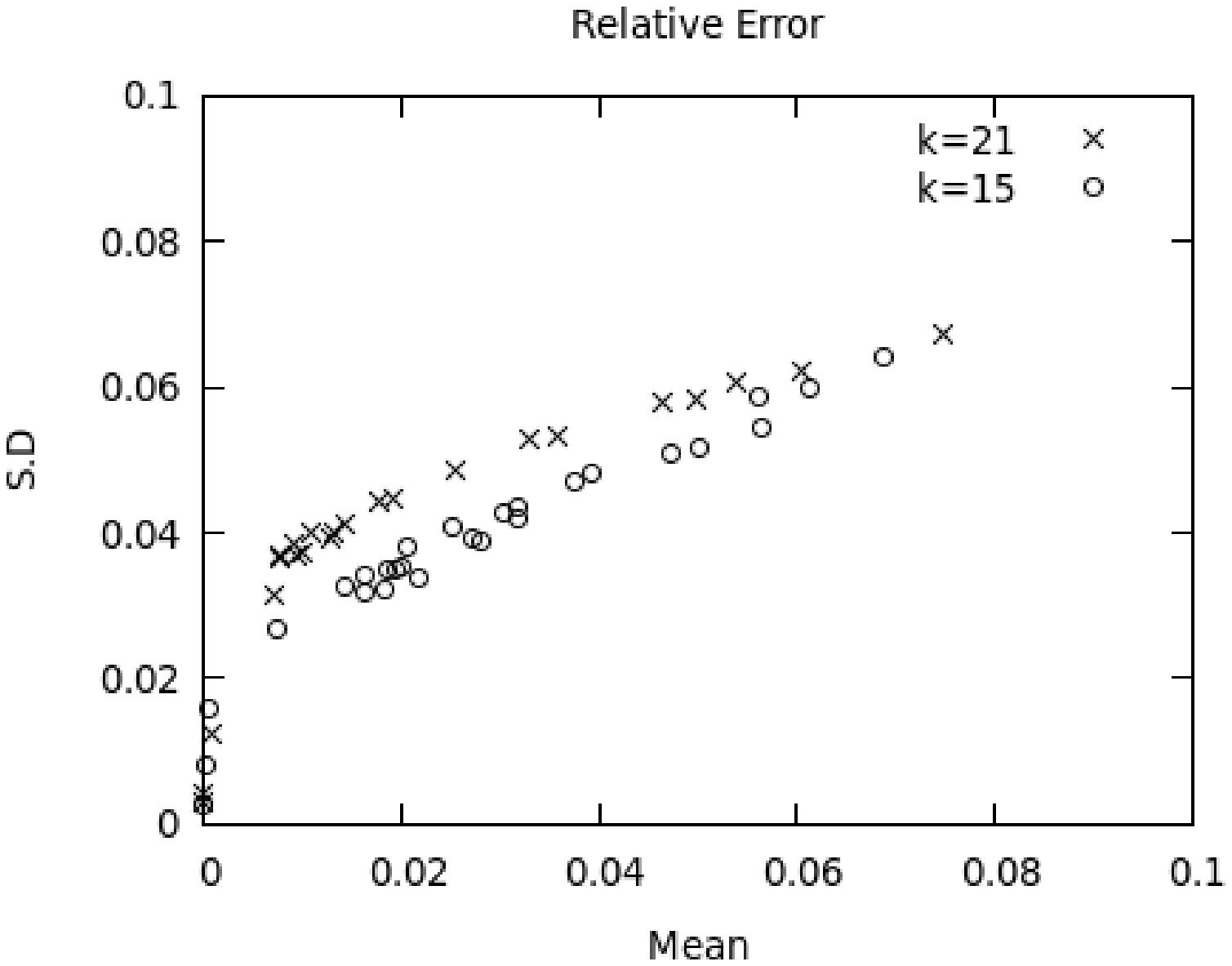}
\end{minipage}
\begin{minipage}[t]{.48\linewidth}
\centering
\includegraphics[height=2.4in]{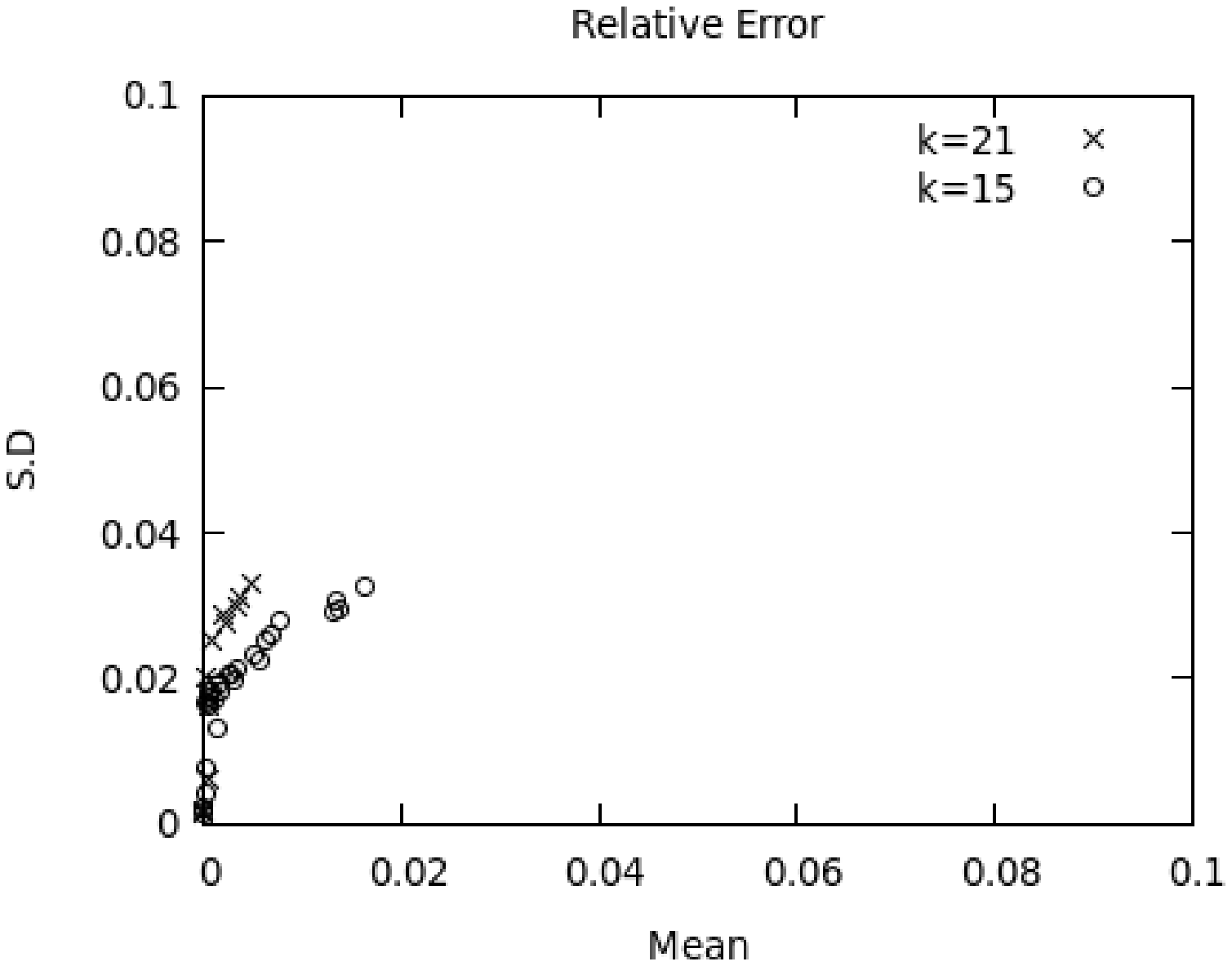}
\end{minipage}
\caption{
Scatter plots for mean and standard deviation of relative errors plotted in Fig \ref{figure:lambda}.  Left and right plots correspond to $r=16$ (120MB) and $r=18$ (500MB)  memory settings respectively for Kmerlight. Plots for both $k=15$ and $k=21$ are provided in each.}
\label{figure:lambda-scatter}
\centering
\end{figure}

\newpage
\section{Estimation of $k$-mer Repeats in Chromosome Y}

Table \ref{tab:ychrome} give frequency statistics of $15$-mers, viz. $g_1, \ldots, g_3$, of human reference chromosome Y (GRCh38) obtained using exact counting.  
There were in total $13493236$ distinct $15$-mers in the Y chromosome. 

\begin{table}[htp]
\centering
\begin{tabular}{| l | c |}
  \hline			
  Freq. $i$ & $g_i$ \\
\hline
\hline
1 &	10076349 \\\hline
2 &	1983449 \\\hline
3 &	526740\\\hline
\end{tabular}
\caption{$k$-mer repeat statistics ($k=15$) of chromosome Y.}
\label{tab:ychrome}
\end{table}

Reads from chromosome Y were generated using ART tool with read length $100$ and coverage value $50$. 
The built-in profile of Illumina HiSeq 2500 system was used by the ART tool.
Table \ref{tab:ychromestat} gives the read statistics.

\begin{table}[htp]
\centering
\begin{tabular}{| l | c |}
\hline
   Reads                 &    11817866\\
\hline
   $k$-mers              &  1016336476\\
\hline
   Distinct $k$-mers       &     75712987\\
\hline

\end{tabular}
\caption{Read collection statistics.}
\label{tab:ychromestat}
\end{table}

Exact $k$-mer abundance histogram of the reads was computed using exact counting. 
Second and third columns of Table \ref{tab:exacthist} gives the first three peak positions and corresponding values observed in the exact histogram, ignoring the initial peak due to erroneous $k$-mers. For these peak positions, the peak values were also inferred using the model, which are given by $g_i/\sqrt{2\pi i \lambda'}$. The $g_i$ values from Table \ref{tab:ychrome} were used. Value of $\lambda'$ was obtained from peak positions since peak positions are given by $i \lambda'$ according to the model.
Estimated peak values and the relative errors are given in third and last columns respectively of  Table \ref{tab:exacthist}.
Small relative errors indicate that the model agrees well with the observed histogram. 

\begin{table}[htp]
\centering
\begin{tabular}{| l | c | c | c | c | }
  \hline			
  Peak No. & Peak position&  Peak value & Inferred & Rel. Error \\
   &  in the &  in the  &  peak value  & \\
   &  histogram&  histogram &    & \\
\hline
\hline
1st & 39 & 629675 & 643696 &  0.022 \\\hline
2nd & 78 & 87496 & 89595 &  0.024\\\hline
3rd & 117 & 18881 & 19427 & 0.029\\\hline
\end{tabular}
\caption{Observed and inferred peak values in the $k$-mer abundance histogram ($k=15$) for human chromosome Y.}
\label{tab:exacthist}
\end{table}

Since computing exact histograms are resource intensive, we used histograms computed by Kmerlight in place of the exact histogram for the inferencing of $g_i$ values as given above.
Figure \ref{fig:kmlhistcoll} shows the abundance histogram computed by Kmerlight alongside the exact histogram. 
Kmerlight histograms computed with three different memory settings viz., 120 MB, 460 MB and 960 MB, are shown.
Values for $f_5, \ldots, f_{200}$ are  plotted in each histogram. Plot of $f_1, \ldots, f_4$ values, which are part of the initial sharp peak due to erroneous $k$-mers,  are omitted from the plots because including them require a very large y-axis range which reduces the resolution of the remaining plot.

\begin{figure}[htbp]
\begin{minipage}[t]{.48\linewidth}
\centering
\includegraphics[height=2.1in]{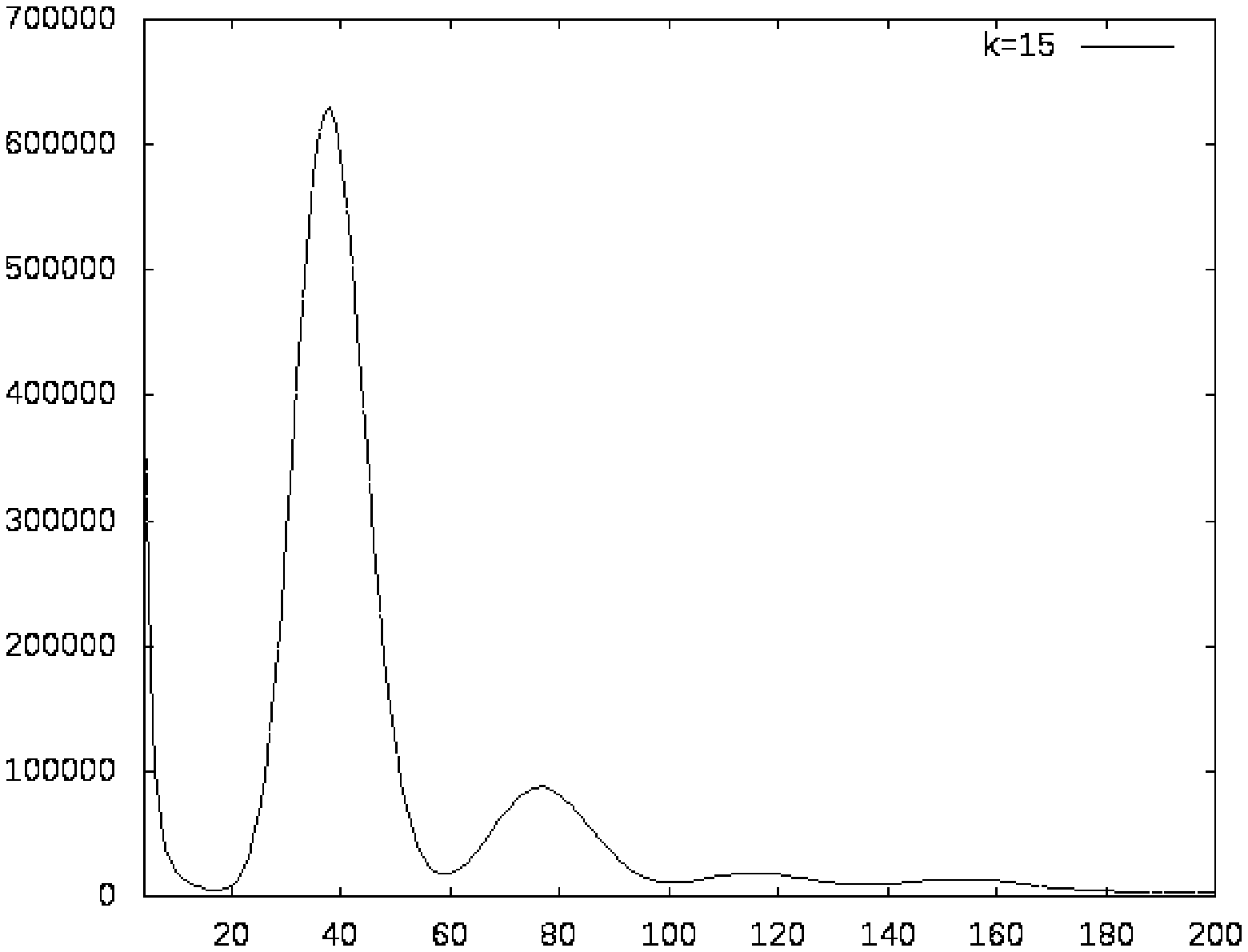}
\end{minipage}
\begin{minipage}[t]{.48\linewidth}
\centering
\includegraphics[height=2.1in]{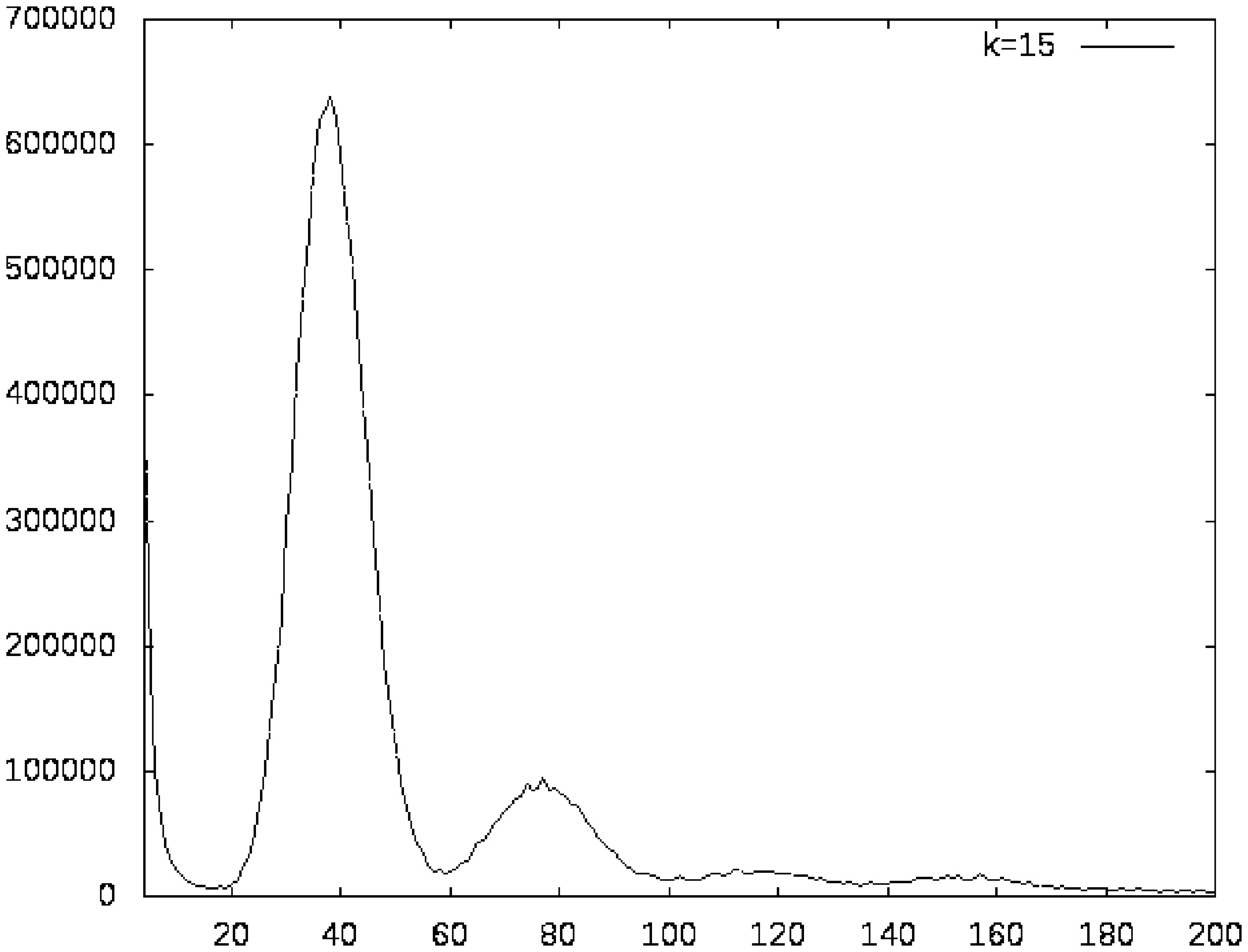}
\end{minipage}

\begin{minipage}[b]{.48\linewidth}
\centering
\includegraphics[height=2.1in]{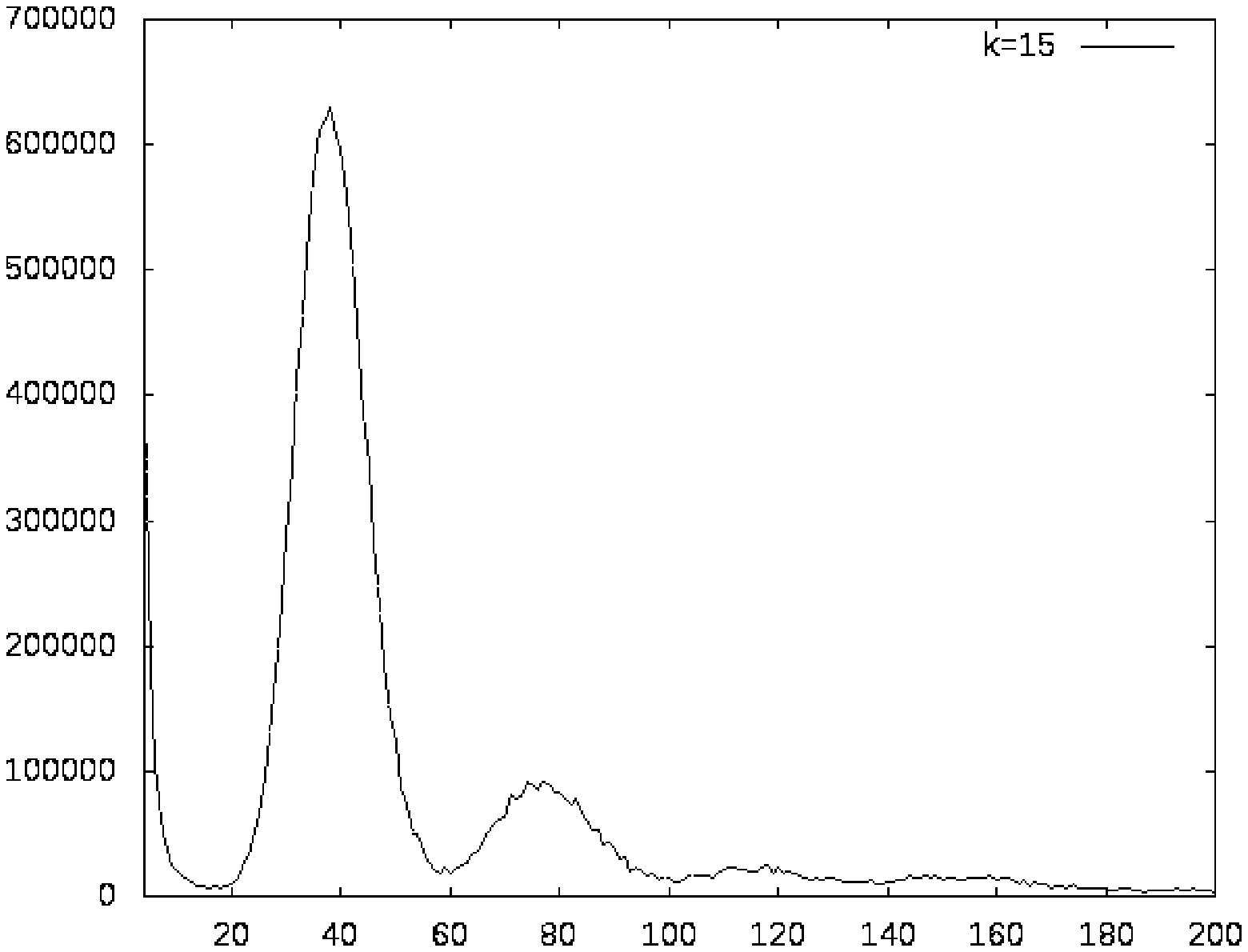}
\end{minipage}
\begin{minipage}[b]{.48\linewidth}
\centering{
\includegraphics[height=2.1in]{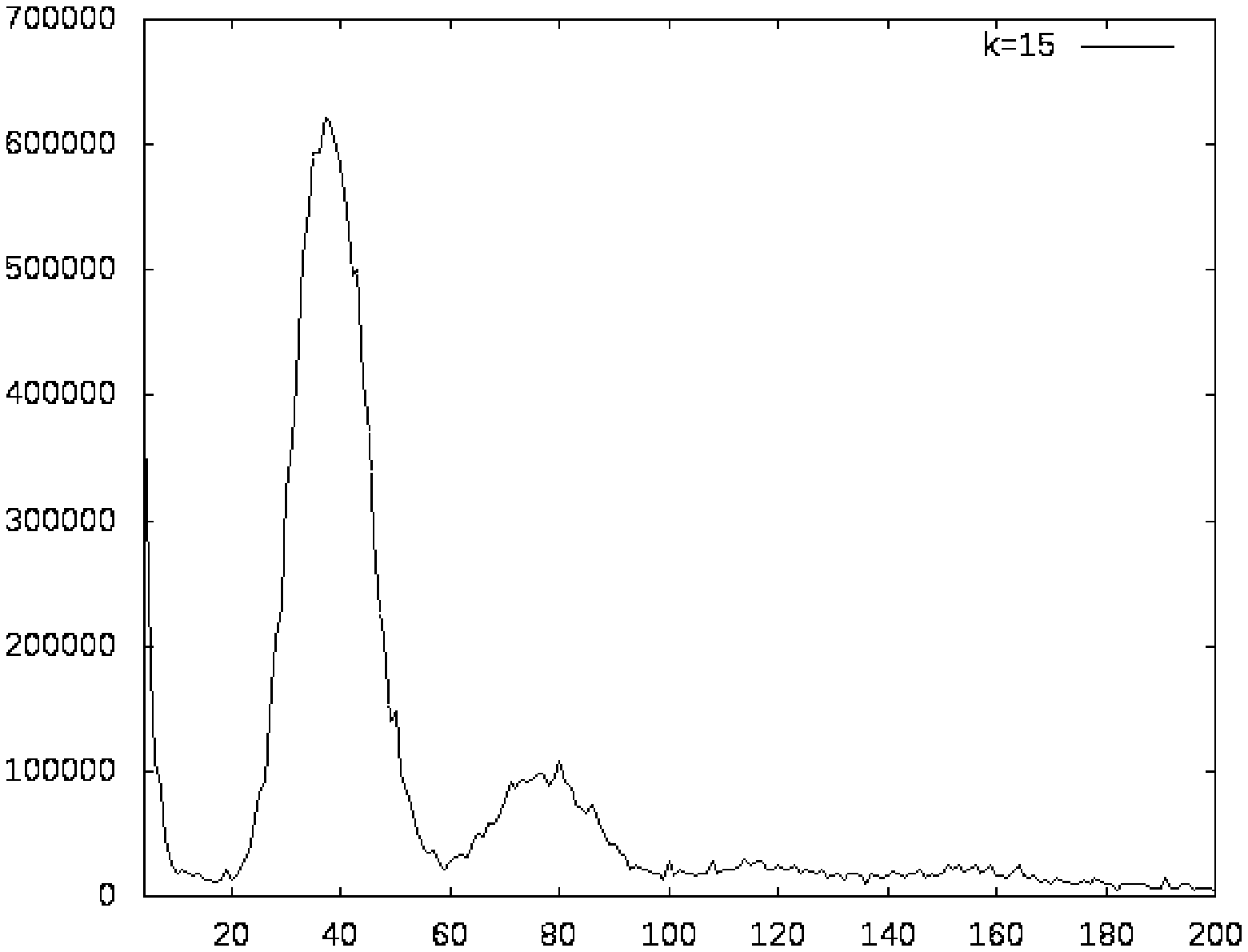}}
\end{minipage}
\caption{$k$-mer abundance histogram for reads from human chromosome Y for $k=15$.  Read length is $100$ and read coverage is $50$. Values of $f_i$ for $i=5, \ldots, 200$ are plotted. Exact histogram is given in top left. Kmerlight histograms for memory settings 960MB, 460MB and 120MB are given in top right, bottom left and bottom right respectively.}
\label{fig:kmlhistcoll}
\end{figure}

Table \ref{tab:gvals1} compares the relative errors in the $g_i$ estimation using exact histogram and using histograms computed by Kmerlight. Third and fourth columns of Table \ref{tab:gvals} give the relative errors in the estimation using histograms computed using Kmerlight with 500 MB and 940 MB memory respectively. 
Kmerlight errors are averaged over 1000 trials.
The standard deviations are also given alongside these values. As seen in Table \ref{tab:gvals1}, relative errors in $g_i$ estimation are similar when histogram computed by Kmerlight is used in place of exact histogram.

\begin{table}[htp]
\centering
\begin{tabular}{| l | c | c | c | }
  \hline			
  $g$ &  Rel. Error  & Rel. Error  & Rel Error \\
   &  in estimation  & in estimation & in estimation  \\
   & (exact histogram) &  Kmerlight (500MB) & Kmerlight (940 MB) \\
   &  &   (S.D.) &  (S.D)\\
\hline
\hline
$g_1$ & 0.02  &  0.02 (0.009) & 0.02 (0.011)\\\hline
$g_2$ & 0.02 &  0.042 (0.028)  & 0.02 (0.015) \\\hline
$g_3$ & 0.028 &  0.208 (0.062) & 0.11 (0.043) \\\hline
\end{tabular}
\caption{Relative errors in $g_i$ estimation using exact histogram and Kmerlight computed histograms with 500 MB and 940 MB memory settings.
Standard deviations are given inside the brackets.}
\label{tab:gvals1}
\end{table}

\newpage

\section{Choosing $k$-mer Length in Genome Assembly}
For choosing appropriate $k$-mer length in de Bruijn based assemblers, Chikhi  et al. \cite{chikhi2013Genie} propose that the most appropriate choice of $k$ is the one that provides maximum number of distinct true $k$-mers to the assembler. This has been shown to yield excellent assembly on a diverse set of genomes. 
Recalling that $F_0$ denotes the total number of distinct $k$-mers in the sequence reads, $F_0$ can be written as $F_0 = F^e_0 + F'_0$, where $F^e_0$ and $F'_0$ are the total number of distinct erroneous $k$-mers and true $k$-mers respectively. Thus, the objective is to choose $k$ that maximizes $F'_0$. 
Estimation of $F'_0$ can be done using estimates of $F_0$ and $F^e_0$ as $F'_0 = F_0 - F^e_0$.  
Summing the $f_i$ values in the initial histogram peak due to erroneous $k$-mers provide an estimate of $F^e_0$. This along with $F_0$ estimate from Kmerlight can be used to estimate $F'_0$.  There could be overlap between the region of erroneous $k$-mers and true $k$-mers in the histogram, which for instance could happen in case of low coverage read generation. This can introduce error in $F^e_0$ estimation obtained by summing  the initial $f_i$ values. In this case, either by using regression techniques or by fitting models for erroneous $k$-mer generation, $f_i$ values corresponding to only erroneous $k$-mers in the overlap region can be inferred.

\section{Proof of Theorem 1}
\vspace*{0.5cm}


\ignore{
	\begin{theorem}
	\label{thm:1}
	With probability at least $1 - \delta$, estimate $\hat{F}_0$ for $F_0$ and estimates $\hat{f}_i$ for every $f_i$ with $f_i \ge F_0/\lambda$ can be computed 
	such that $(1-\epsilon) f_i \le \hat{f}_i \le (1+\epsilon) f_i$ for all such $f_i$, using
	$O(\frac{\lambda}{\epsilon^2} \log(\lambda/\delta) \log(F_0))$ memory and $O(\log(\lambda/\delta))$ time per update.
	\end{theorem}
}

\begin{theorem*}
With probability at least $1 - \delta$, estimate $\hat{F}_0$ for $F_0$ and estimates $\hat{f}_i$ for every $f_i$ with $f_i \ge F_0/\lambda$ can be computed 
such that $(1-\epsilon) F_0 \le \hat{F}_0 \le (1+\epsilon) F_0$ and  $(1-\epsilon) f_i \le \hat{f}_i \le (1+\epsilon) f_i$ using
$O(\frac{\lambda}{\epsilon^2} \log(\lambda/\delta) \log(F_0))$ memory and $O(\log(\lambda/\delta))$ time per update.
\end{theorem*}

We make the simplifying assumption that all hash functions used are truly random. We refer the reader to \cite{mitzenmacher2008simple}
where it is shown that standard hash families such as $2$-universal hash families `well approximates' ideal hash functions when the data entropy is high. In this case, the probabilistic bounds obtained using ideal hash functions are `close' to the bounds achievable using the universal hash families. We assume that parameters are suitably adjusted in such a way that the final failure probabilities accommodate for the additional difference arising due to use of universal hash families in practice.

Fix any $f_k$ for $k \ge 1$ such that $f_k \ge F_0/\lambda$. Consider the Kmerlight data structure and consider any fixed array $T_w$  corresponding to level $w \ge 1$. We recall that counters of array $T_w$ can be either `dirty' or `non dirty'. 
Our algorithm considers only non dirty counters.
The set of non dirty counters can contain both false positive counters and true positive counters. 
 However, the algorithm cannot distinguish between these false positive and true positive counters. If a false positive counter has value $k$ then it can erroneously contribute to $f_k$ calculation.  Let $X_w$ denote the number of true positive counters in array $T_w$ each having value $k$ after seeing the input. We will later bound the estimation error due to false positive counters. 

\begin{claim}\label{cl:expvar}
$E(X_w) = \frac{f_k}{2^w} \left( 1 - \frac{1}{r \cdot 2^w } \right)^{F_0 - 1}$ 
and $Var(X_w) \le E(X_w)$.
\end{claim}
\begin{proof}
Let $x_1, x_2, \ldots, x_r$ denote $r$ indicator random variables where $x_i = 0$ if the counter $i$ in array $T_w$ is true positive and holds value $k$. Thus $X_w = \sum_{i=1}^r x_i$.
For any $i \in \{1, \ldots, r\}$, we have
\begin{eqnarray*}
E(x_i) = \frac{f_k}{2^w \cdot r} \left( 1 - \frac{1}{r \cdot 2^w } \right)^{F_0 - 1}
\end{eqnarray*}
Thus
\begin{eqnarray*}
E(X_w) = \frac{f_k}{2^w} \left( 1 - \frac{1}{r \cdot 2^w } \right)^{F_0 - 1}
\end{eqnarray*}

We use the relation $Var(X_w) = E(X_w^2) - E^2(X_w)$. We bound $E(X_w^2)$ as
\begin{eqnarray*}
E(X_w^2) & = & E((x_1 + \ldots + x_r)^2) \nonumber\\
	& = & \sum_{i=1}^r E(x_i^2) + \sum_{i \not= j} E(x_i x_j) \nonumber\\
	& = & E(X_w) + \sum_{i \not= j} E(x_i x_j) \nonumber\\
	& = & E(X_w) + \sum_{i \not= j} \Pr(x_i = 1 \wedge x_j = 1) \nonumber\\
	& = & E(X_w) + \sum_{i \not= j} {f_k \choose 2} \left(\frac{1}{r \cdot 2^w }\right)^2 \left(1 - \frac{2}{r \cdot 2^w}\right)^{F_0 - 2} \nonumber\\
	& \le & E(X_w) + r^2 {f_k \choose 2} \left(\frac{1}{r \cdot 2^w}\right)^2 \left(1 - \frac{2}{r \cdot 2^w}\right)^{F_0 - 2} \nonumber\\
	& \le & E(X_w) + \left(\frac{1}{2}\right)  \left(\frac{f_k}{2^w}\right)^2 \left(1 - \frac{2}{r \cdot 2^w}\right)^{F_0 - 2} \nonumber\\
	& \le & E(X_w) + \frac{1}{2}\left(\frac{r}{r-1}\right)  \left(\frac{f_k}{2^w}\right)^2 \left(1 - \frac{2}{r \cdot 2^w}\right)^{F_0 - 1} \nonumber\\
	& \le & E(X_w) + \left(\frac{f_k}{2^w}\right)^2 \left(1 - \frac{1}{r \cdot 2^w}\right)^{2(F_0 - 1)} \nonumber
\end{eqnarray*}
The last inequality follows from facts $r/2(r-1) \le 1$ for $r \ge 2$ and $1-2x \le (1-x)^2$. Observing that the second term of the last inequality is $E^2(X_w)$, 
we conclude that $Var(X_w) = E(X_w^2) - E^2(X_w) \le E(X_w)$.
\end{proof}

First we show lower and upper bounds on $E(X_w)$. Let $\theta = r/(r-1)$. We assume $r \ge 2$.  Let
$$
L(w) =  \frac{f_k} {2^{w}} \exp\left( \frac{-\theta F_0}{r \cdot 2^{w}}\right) \mbox{ ~~and~~ }  H(w) =  \frac{f_k} {2^{w}} \cdot \theta \cdot \exp\left( \frac{-F_0}{r \cdot 2^{w}}\right)
$$

Using Claim \ref{cl:expvar}, we obtain
\begin{eqnarray} 
E(X_{w}) & \ge &  \frac{f_k} {2^{w}} \left( 1 - \frac{1}{r \cdot 2^{w}}\right)^{F_0}  \nonumber \\
	 & \ge &  \frac{f_k} {2^{w}} \exp\left( \frac{-\theta F_0}{r \cdot 2^{w}}\right) ~=~ L(w) \nonumber
\end{eqnarray}
The last inequality follows from the inequality 
\begin{eqnarray}
\label{eq:alpha}
1-x \ge e^{-t x}   ~~\mbox{   for   }~~  x \in [0, 1] \mbox{~and~} x \le 2(t -1)/t^2
\end{eqnarray}

Similarly, we obtain
\begin{eqnarray} 
E(X_{w}) & = &  \frac{f_k} {2^{w}} \left( 1 - \frac{1}{r \cdot 2^{w}}\right)^{F_0-1}  \nonumber \\
	 & \le &  \frac{f_k} {2^{w}} \cdot \theta  \cdot \exp\left( \frac{-F_0}{r \cdot 2^{w}}\right) ~=~ H(w) \nonumber
\end{eqnarray}
where the last inequality follows from Claim (\ref{cl:expvar}), the definition of $\theta$, and from the inequality 
\begin{eqnarray}
\label{eq:exp}
1+x \le \exp(x) ~~\mbox{ for any real } x
\end{eqnarray}

It follows that 
\begin{eqnarray}
\label{eq:exbounds}
L(w) ~\le~ E(X_w) ~\le~ H(w)
\end{eqnarray}

Consider the interval $[l, u]$, where $l$ is the smallest level with $F_0/2^l \le 4r$ and $u$ is the largest level with $F_0/2^u \ge r/8$. That is,
\begin{eqnarray}
\label{eq:lbound}
2r ~~<~~ F_0/2^l ~\le~ 4r \\
\label{eq:ubound}
r/8 ~~\le~~ F_0/2^{u}  ~<~  r/4
\end{eqnarray}

From the above two inequalities, it follows that 
\begin{equation}
\label{eq:u-l}
u - l \le 5
\end{equation}

\begin{claim}
\label{cl:geo}
For $w \in [1, l-2]$, $H(w) \le H(w+1)/2$ and $\sum_{w=1}^{l-1} H(w) \le \frac{27}{26}H(l-1)$.  Similarly, for $w \ge u+1$, $H(w+1) \le 2H(w)/3$ and $\sum_{w \ge u+1} H(w) \le \frac{7}{3}H(u+1)$.
\end{claim}
\begin{proof}
For any $w \le l-2$, using (\ref{eq:lbound}), we obtain 
$$
\frac{H(w+1)}{H(w)} = \frac{1}{2} \cdot \exp\left(\frac{F_0}{r \cdot 2^{w+1}}\right) \ge \frac{1}{2} \cdot \exp\left(\frac{F_0}{r \cdot 2^{l-1}}\right) \ge e^4/2  \ge 27 
$$
Hence, $\sum_{w=1}^{l-1} H(w) ~\le~ H(l-1) \sum_{w\ge 0} (1/27)^{w} ~\le~ 27H(l-1)/26$.

Similarly, for any $w \ge u+1$, using (\ref{eq:ubound}), we obtain
$$
\frac{H(w)}{H(w+1)} = 2 \cdot \exp\left(-\frac{F_0}{r \cdot 2^{w+1}}\right) \ge 2 \cdot \exp\left(-\frac{F_0}{r \cdot 2^{u+2}}\right) \ge 2e^{-1/16}  \ge  7/4
$$
Hence, $\sum_{w \ge u+1} H(w) ~\le~ H(u+1) \sum_{w \ge 0} (4/7)^w ~\le~ 7H(u+1)/3$.
\end{proof}

We assume that $\epsilon < 0.1$.

\begin{lemma}
\label{lem:good}
When $r \ge 1400 \cdot \lambda/\epsilon^2$, probability that there exists a level $w \in [l, u]$ with $|X_w - E(X_w)| \ge \epsilon E(X_w)$ is at most $1/8$.
\end{lemma}

\begin{proof}

First we show that for any given level $w \in [l, u]$, $|X_w - E(X_w)| \ge \epsilon E(X_w)$ with probability at most $1/48$. 
The result then follows from (\ref{eq:u-l}) by noting that there are at most $6$ levels in $[l, u]$.

For $w \in [l, u]$, using (\ref{eq:exbounds}) and (\ref{eq:lbound}),  we can write 
\begin{eqnarray}
\frac{1}{E(X_w)} & \le & \frac{1}{L(w)} \nonumber \\ 
		& = & \left( \frac{2^w}{f_k}\right)\exp\left(\frac{\theta F_0}{r \cdot 2^w}\right) \nonumber \\
		& \le & \left( \frac{2^w}{f_k}\right)\exp(4 \theta/2^{w-l}) \nonumber \\
		& \le & \left( \frac{\lambda}{2r}\right)2^{w-l} \cdot \exp(4 \theta/2^{w-l}) \nonumber \\
		\label{eq:expinv}
		& \le & \left( \frac{\lambda}{2r}\right) \exp(4 \theta) 
\end{eqnarray}

The second last inequality follows from the assumption that $f_k \ge F_0 / \lambda$. The last inequality can be easily verified for all possible values for $w-l$
which by (\ref{eq:u-l}) is given by the set $\{0, \ldots, 5\}$ and also noting that $\theta \ge 1$.

We apply Chebyshev inequality and use the above upper bound for $1/E(X_w)$ and the fact that $Var(X_w) \le E(X_w)$ due to Claim \ref{cl:expvar} to obtain
\begin{eqnarray*}
\Pr[|X_w - E(X_w)| \ge \epsilon E(X_w)] & \le & Var(X_w)/\epsilon^2 E^2(X_w) \\
& \le & \frac{1}{\epsilon^2 E(X_w)}\\
& \le & \frac{1}{\epsilon^2} \left(\frac{\lambda}{2r}\right) \exp(4\theta)\\
& \le & 1/48
\end{eqnarray*}
 where the last inequality follows because $r \ge 1400 \cdot \lambda /\epsilon^2$. 
\end{proof}

Consider level $w' \in [l, u]$ such that $r/2 ~\le~ F_0/2^{w'} ~\le~ r$. From (\ref{eq:lbound}) and (\ref{eq:ubound}), it is clear that such a $w'$ exists.

\begin{claim}
\label{cl:lexpclaim}
$(1+\epsilon) H(l-1) < (1-\epsilon) E(X_{w'})$.
\end{claim}
\begin{proof}
From (\ref{eq:exbounds}), we obtain
\begin{eqnarray}
\label{eq:wstar}
E(X_{w'}) ~\ge~ L(w') ~=~ f_k/2^{w'} \exp\left( \frac{-\theta F_0}{r \cdot 2^{w'}}\right) \ge \left(\frac{r f_k}{2 F_0} \right)  \exp(-\theta)
\end{eqnarray}
Using (\ref{eq:lbound}), we obtain
$$
H(l-1) ~=~  \frac{f_k} {2^{l-1}} \cdot \theta \cdot \exp\left( \frac{-F_0}{r \cdot 2^{l-1}}\right) ~\le~  \frac{8r f_k} {F_0} \cdot \theta \cdot e^{-4}
$$
It is straightforward to verify that the result follows from the above two bounds, for $\epsilon < 0.1$ and $r \ge 101$.
\end{proof}

\begin{claim}
\label{cl:uexpclaim}
$(1+\epsilon) H(u+1) < (1-\epsilon) E(X_{w'})$.
\end{claim}
\begin{proof}
From (\ref{eq:wstar}), we have
$$
E(X_{w'}) ~\ge~ \left(\frac{r f_k}{2 F_0} \right) \exp(-\theta)
$$
Using (\ref{eq:ubound}), we obtain
$$
H(u+1) ~=~  \frac{f_k} {2^{u+1}} \cdot \theta \cdot \exp\left( \frac{-F_0}{r \cdot 2^{u+1}}\right) ~\le~  \frac{r f_k} {8 F_0} \cdot \theta \cdot e^{-1/16}.
$$
Again, it is straightforward to verify that the result follows from the above two bounds, for $\epsilon < 0.1$ and $r \ge 101$.
\end{proof}

\begin{lemma}
\label{lem:badrows}
Probability that there exists a level $w$ outside $[u, l]$ with $X_w \ge (1-\epsilon)E(X_{w'})$ is at most $1/2$ when $r \ge 2000  \lambda/\epsilon^2$.
\end{lemma}
\begin{proof}
Consider any level $w \in [1, l-1]$. From Claim \ref{cl:geo} we have $H(w) \le H(l-1)$ and from Claim \ref{cl:lexpclaim} we have $H(l-1) < (1-\epsilon) E(X_{w'})$. Using (\ref{eq:exbounds}), it follows that $E(X_w)  \le H(w) \le H(l-1) < (1 - \epsilon) E(X_{w'})$. Using this and using Chebyshev inequality we obtain

\begin{eqnarray*}
\Pr[X_w \ge (1-\epsilon) E(X_{w'})] & \le & \Pr[|X_w - E(X_w)| \ge (1 -\epsilon)E(X_{w'}) - E(X_w)] \\
	& \le & Var(X_w) / ((1 -\epsilon)E(X_{w'}) - E(X_w))^2 \\
	& \le & Var(X_w) / ((1 -\epsilon)E(X_{w'}) - H(l-1))^2  \\ 
	& \le & Var(X_w) / (\epsilon H(l-1))^2 ~~(\mbox{by Claim \ref{cl:lexpclaim}})\\
	& \le & E(X_w) / (\epsilon H(l-1))^2 ~~(\mbox{by Claim \ref{cl:expvar}})\\
	& \le & \frac{H(w)}{(\epsilon H(l-1))^2}
\end{eqnarray*} 

Similarly, for any level $w \ge u+1$, we know from Claim \ref{cl:geo} and  Claim \ref{cl:uexpclaim} that $E(X_w) \le H(w) \le H(u+1) < (1 - \epsilon) E(X_{w'})$. Hence by Chebyshev inequality,
\begin{eqnarray*}
\Pr[X_w \ge (1-\epsilon) E(X_{w'})] & \le & \Pr[|X_w - E(X_w)| \ge (1 -\epsilon)E(X_{w'}) - E(X_w)] \\
	& \le & Var(X_w) / ((1 -\epsilon)E(X_{w'}) - E(X_w))^2 \\  
	& \le & Var(X_w) / ((1 -\epsilon)E(X_{w'}) - H(u+1))^2 \\ 
	& \le & Var(X_w) / (\epsilon H(u+1))^2 ~~(\mbox{by Claim \ref{cl:uexpclaim}})\\
	& \le & E(X_w) / (\epsilon H(u+1))^2\\
	& \le & \frac{H(w)}{(\epsilon  H(u+1))^2}
\end{eqnarray*} 

Now using Claim \ref{cl:geo}, we obtain that the  probability that there exists a level $w$ outside $[u, l]$ with $X_w \ge (1-\epsilon)E(X_{w'})$ is at most 
\begin{eqnarray*}
& & \sum_{i=1}^{l-1} \frac{H_w}{(\epsilon H(l-1))^2}  + \sum_{i\ge u+1} \frac{H_w}{(\epsilon H(u+1))^2} \\
 & \le & \frac{1}{\epsilon^2} \left(\frac{27}{26} \cdot \frac{1}{H(l-1)} + \frac{7}{3}\cdot \frac{1}{H(u+1)}\right) \\ 
 & \le & \frac{1}{\epsilon^2} \left(\frac{27}{26} \cdot \frac{2^{l-1}}{f_k} \cdot \exp\left(\frac{F_0}{r 2^{l-1}}\right)  + \frac{7}{3}\cdot \frac{2^{u+1}}{f_k} \cdot \exp\left(\frac{F_0}{r 2^{u+1}}\right)\right) \\
 & \le & \frac{\lambda}{\epsilon^2} \left(\frac{27}{26} \cdot \frac{2^{l-1}}{F_0} \cdot \exp\left(\frac{F_0}{r 2^{l-1}}\right)  + \frac{7}{3}\cdot \frac{2^{u+1}}{F_0} \cdot \exp\left(\frac{F_0}{r 2^{u+1}}\right)\right) \\
 & \le & \frac{\lambda}{\epsilon^2 r} \left(\frac{27}{104} \exp(8) + \frac{112}{3}  \exp(1/8)\right)  \mbox{~(by (\ref{eq:lbound}) and (\ref{eq:ubound}))}\\
& \le & \frac{817 \lambda}{\epsilon^2 r} \le 1/2
\end{eqnarray*}
for $r \ge 2000 \cdot \lambda/\epsilon^2$.
\end{proof}

\begin{lemma}
\label{lem:finalnofp}
The level $w^*$ computed by the algorithm is such that $|X_{w^*} - E(X_{w^*})| \le \epsilon E(X_{w^*})$ with probability at least $3/8$ when $r \ge 2000 \lambda/\epsilon^2$.
\end{lemma}

\begin{proof}
From Lemma \ref{lem:good}, and Lemma \ref{lem:badrows}, it follows that with probability at least $3/8$, the following properties holds.  For any $w\in [l, u]$,  $|X_w - E(X_w)| < \epsilon E(X_w)$. In particular, it holds for $w' \in [l, u]$. That is, $(1 - \epsilon) E(X_{w'}) < X_{w'} < (1+ \epsilon) E(X_{w'})$. Furthermore, for any $w$ outside $[l, u]$, $X_w < (1- \epsilon) E(X_{w'})$. 
We recall that the level $w^*$ computed by the algorithm is the level which maximizes the total number of non dirty counters in that level. Non dirty counters could contain both false positive and true positive counters. We assume for now that the level $w^*$ computed by the algorithm maximizes the number of true positive counters, which is $X_w$. We will remove this assumption in the next section. That is, $w^* = \arg\!\max_w \{X_w\}$. 
As a consequence, it follows that with probability at least $3/8$, the level $w^*$ lies in the range $[l, u]$ and 
it satisfies the property that $|X_{w^*} - E(X_{w^*})| \le \epsilon E(X_{w^*})$.
\end{proof}

\newpage
\vspace*{0.2cm}
\noindent
\textbf {Dealing with false positives}
\vspace*{0.2cm}

Consider any level $w$ and a fixed array location $T_w[c]$ of array $T_w$ in level $w$. Let $f_p$ denote the probability that the counter value stored in $T_w[c]$ is false positive.  Hence $T_w[c].v$ value is erroneously considered towards $f_k$ estimation for $k \ge 2$. We recall that false positives do not arise in $f_1$ estimation. 

\begin{lemma}
\label{lem:fp}
$f_p \le 1/u$ for $r \ge 2$ and $u \ge 6$.
\end{lemma}

\begin{proof}
As earlier, we assume that hash functions are truly random.
Fix $j \in \{0, \ldots, u-1\}$.
Consider the following two events:

\noindent
Event $E_1$: Two or more elements are hashed to  the pair $(c, j)$ in level $w$

\noindent
Event $E_2$: No elements are hashed to pairs $(c, i)$ for any $i \not= j$ in level $w$.

We observe that $f_p \le u \cdot \Pr( E_1 \wedge E_2 )$. In the following, we show that $\Pr (E_1 \wedge E_2) \le 1/u^2$, which would then imply that $f_p \le 1/u$. We can analyze $\Pr( E_1 \wedge E_2)$ by considering the following balls and bins experiment where $F_0 = n$ balls are thrown into $u$ bins in the following fashion. For each ball, a coin is tossed with success probability $p = 1/(r \cdot 2^w)$. With probability $1-p$ the ball is discarded. Upon success, the ball is thrown into any one  of the $u$ bins $\{0, \ldots, u-1\}$  with equal probability. It is straightforward to verify that event $E_1 \wedge E_2$ is equivalent to the event that two or more balls are present in bin $j$ and remaining $u-1$ bins are empty. The probability of this event is given by:

\begin{eqnarray*}
	\Pr(E_1 \wedge E_2) & = &\sum_{k=2}^n {n \choose k} (p/u)^k (1-p)^{n-k}\\
	& = &(1-p)^n \sum_{k=2}^n {n \choose k} \left(\frac{p}{u(1-p)}\right)^k \\
	& \le & \exp(-np) \sum_{k=2}^n {n \choose k} \left(\frac{p}{u(1-p)}\right)^k \\
	& \le &  \exp(-np) \sum_{k=2}^n \frac{1}{k!} \left(\frac{np}{u(1-p)}\right)^k \\
	& \le & \left(\frac{1}{2}\right)  \exp(-np) \left(\frac{np}{u(1-p)}\right)^2 \sum_{k\ge 0} \frac{1}{k!} \left(\frac{np}{u(1-p)}\right)^k \\
	& \le & \left(\frac{1}{2}\right) \exp(-np) \left(\frac{np}{u(1-p)}\right)^2 \exp\left(\frac{np}{u(1-p)}\right) \\
	& = & \left(\frac{1}{2}\right) \left(\frac{np}{u(1-p)}\right)^2 \exp\left\{-np\left(1 - \frac{1}{u(1-p)}\right)\right\} \\
\end{eqnarray*}

We use the facts that $p \le \frac{1}{2r}$ and $r \ge 2$ to obtain that $1/(1-p)^2 \le 2$. Using this, we simplify the above expression as
\begin{eqnarray*}
	\Pr(E_1 \wedge E_2) &  \le  & \left(\frac{np}{u}\right)^2 \exp\left\{-np\left(1 - \frac{2}{u}\right)\right\} \\
	&  =  & \frac{1}{u^2} (np)^2 \exp\left\{-np\left(1 - \frac{2}{u}\right)\right\} \\
	&  \le  & \frac{1}{u^2} (np)^2 \exp(-3np/4) ~~(\mbox{for~} u \ge 8)\\
	& \le & 1/u^2,
\end{eqnarray*}
for $u \ge 8$. The last inequality follows because $x^2 \le \exp(3x/4)$ for all $x\ge 0$.
\end{proof}

For the fixed $k$, let $Y_w$ denote the number of non dirty locations in the array $T_w$  each having value $k$ at level $w$. Let $Y_w = X_w + \Delta_w$, where $X_w$ as defined earlier is the number true positive locations and $\Delta_w$ is the number of false positive locations.

\ignore{ 
	Let $\Delta_w$ denote the number of false positive array locations in row $w$. From Lemma \ref{lem:fp}, we have $E(\Delta_w) \le r/u$. Hence by Markov inequality,
	\begin{eqnarray}
	\Pr[\Delta_w \ge b] \le \frac{r}{ub}
	\end{eqnarray}
}


From Lemma \ref{lem:fp}, for any level $w$, we have, 
$$
E(\Delta_w) \le r/u.
$$

Let $\epsilon' = \epsilon - \epsilon^2$. 

 We recall the definition of level $w'$. By Markov inequality, it follows that for any fixed $w$ outside $[l, u]$,
\begin{eqnarray*}
\Pr[\Delta_{w}  \ge  (\epsilon - \epsilon') E(X_{w'})] & \le &  \frac{r}{u (\epsilon - \epsilon') E(X_{w'})}\\
   & \le & \frac{2 \lambda e^{\theta}}{u \epsilon^2 } ~~~(\mbox{using (\ref{eq:wstar}}))\\
   & \le & \frac{1}{16 \log(F_0)}
\end{eqnarray*}
for $u =O(\lambda \log(F_0) /\epsilon^2)$.  Since there are at most $\log(F_0)$ levels, it follows that
\begin{equation}\label{eq:prime1}
\Pr\left[\Delta_{w}  \ge  (\epsilon - \epsilon') E(X_{w'}) \mbox{~for any~} w \mbox{~outside~} [l, u]\right]   \le  1/16
\end{equation}
when $u = O(\lambda \log(F_0)/\epsilon^2)$.

\ignore{
	Similarly for any $w \in [l, u]$, 
	\begin{eqnarray*}
	\Pr[\Delta_{w}  \ge  (\epsilon - \epsilon') E(X_{w})] & \le &  \frac{r}{u (\epsilon - \epsilon') E(X_{w})}\\
	\end{eqnarray*}
}


Similarly for any fixed $w \in [l, u]$, by Markov inequality,  we obtain that
\begin{eqnarray*}
\Pr[\Delta_{w}  \ge  (\epsilon - \epsilon') E(X_{w})] & \le &  \frac{r}{u (\epsilon - \epsilon') E(X_w)} \\
& \le &  \frac{r}{u (\epsilon - \epsilon')} \cdot \frac{\lambda}{2r} \cdot e^{4\theta} ~~\mbox{(using (\ref{eq:expinv}))}\\
& = &  \frac{\lambda e^{4\theta}}{2u \epsilon^2}  \le  1/100
\end{eqnarray*}
for $u = O(\lambda/\epsilon^2)$. Recalling that there are at most 6 intervals in $[l, u]$, it follows that 
\begin{equation}\label{eq:prime2}
\Pr\left[\Delta_{w}  \ge  (\epsilon - \epsilon') E(X_{w}) \mbox{~for any~} w \in [l, u]\right]  \le 1/16
\end{equation}
 when  $u = O(\lambda /\epsilon^2)$.

Now from Lemma \ref{lem:good} we have 
\begin{equation}\label{eq:prime3}
\Pr\left[|X_{w} - E(X_w)| \ge \epsilon' E(X_w) \mbox{~for any~} w \in [l, u]\right] \le 1/8
\end{equation}
where  $r = O(\lambda /\epsilon'^2) = O(\lambda/\epsilon^2)$ when $\epsilon < 0.5$.

Similarly from Lemma \ref{lem:badrows} we have 
\begin{equation}\label{eq:prime4}
\Pr\left[X_{w} \ge (1- \epsilon - \epsilon^2) E(X_{w'}) \mbox{~for any~} w \mbox{~outside~} [l, u]\right] \le 1/2
\end{equation}
when  $r = O(\lambda /\epsilon^2)$.

From (\ref{eq:prime1}) (\ref{eq:prime2})  (\ref{eq:prime3}) and (\ref{eq:prime4}), we conclude the following:
with probability at least $1 - 1/8 - 1/2 - 1/16/ - 1/16 = 1/4$, 
$$
(1 - \epsilon) E(X_w) \le Y_w \le (1+\epsilon) E(X_w) \mbox{~ for all ~} w \in [l , u],
$$
and
$$
Y_w < (1 - \epsilon) E(X_{w'}) \mbox{~ for all ~} w \mbox{~outside~} [l , u],
$$
when $r = O(\lambda/\epsilon^2)$ and $u = O(\lambda \log(F_0)/\epsilon^2)$.
Recalling that $w' \in [l, u]$, it follows that the level $w^*$ computed by the algorithm lies in $[l, u]$. As a consequence, we have the following Lemma.

\begin{lemma}
\label{lem:finalwithfp}
The level $w^*$ computed by the algorithm is such that $|Y_{w^*} - E(X_{w^*})| \le \epsilon E(X_{w^*})$ with probability at least $1/4$ when $r = O(\lambda/\epsilon^2)$ and $u= O(\lambda\log(F_0)/\epsilon^2)$.
\end{lemma}


Suppose we have an estimate $\hat{F}_0$ with $|\hat{F}_0 - F_0| \le   \alpha F_0$ and we have $Y_{w^*}$ with $|Y_{w^*} - E(X_{w^*})| \le \beta E(X_{w^*})$. 
Using $\hat{F}_0$ and $Y_{w^*}$, we compute estimate $\hat{f}_k$ for $f_k$ as follows. 
We  recall from Claim \ref{cl:expvar} that $E(X_{w^*}) = \frac{f_k}{2^{w^*}} \left( 1 - \frac{1}{r \cdot 2^{w^*} } \right)^{F_0 - 1}$.
We recall that the estimate $\hat{f}_k$ is given by
$$
\hat{f}_k ~=~  2^{w^*} \cdot Y_{w^*} \left( 1 - \frac{1}{r } \right)^{1 - \frac{\hat{F}_0}{2^{w^*}}}
$$

To bound error for $\hat{f}_k$, we have
$$
\hat{f}_k ~\in~  2^{w^*} \cdot (1 \pm \beta) E(X_{w^*}) \left( 1 - \frac{1}{r  } \right)^{1 - (1 \pm \alpha) \frac{F_0}{2^{w^*}}}
$$

This can be restated as
$$
(1 - \beta) \rho f_k  \left( 1 - \frac{1}{r \cdot 2^{w^*} } \right)^{\alpha F_0} ~\le~ \hat{f}_k ~\le~ (1 + \beta) \rho f_k  \left( 1 - \frac{1}{r \cdot 2^{w^*} } \right)^{-\alpha F_0}
$$
where $\rho = \left( 1 - \frac{1}{r}\right) ^{1 - \frac{F_0}{2^{w^*}}} \left( 1 - \frac{1}{r 2^{w^*}} \right) ^{F_0 - 1}$.

We recall $\theta = r/(r-1)$ and use (\ref{eq:alpha}), (\ref{eq:exp}) and the fact that $F_0/2^{w^*} \le F_0/2^l \le 4r$ from (\ref{eq:lbound}), to bound $\hat{f}_k$ as 

$$
(1-\beta) \rho (1-4\alpha \theta) f_k ~\le~ (1-\beta) \rho \exp\left(\frac{-\theta\alpha F_0}{r2^{w^*}}\right) f_k ~\le~ \hat{f}_k ~\le~  (1+\beta) \rho \exp\left(\frac{\alpha F_0}{r2^{w^*}}\right) f_k ~\le~ \frac{(1+\beta)}{(1 - 4\alpha)} \rho  f_k
$$

We can also bound $\rho$ as

$$
(1 - 1/r) e^{-4/(r-1)} \le (1 - 1/r) \exp\left(-(\theta - 1)\frac{F_0}{r 2^{w^*}}\right) ~\le~ \rho ~\le~ \exp\left((\theta - 1)\frac{F_0}{r 2^{w^*}}\right) \le e^{4/(r-1)},
$$
which using (\ref{eq:exp}) and the fact $r = \Omega(1/\epsilon^2)$, can be simplified as 
$$
1- \epsilon^2 \le \rho \le 1 +  \epsilon^2
$$

It is straightforward  to verify using (\ref{eq:exp}) that for suitable constants  $c_1$ and $c_2$ with $\alpha = c_1 \epsilon$ and $\beta = c_2 \epsilon$, we have $(1-\epsilon) \le (1-\beta) \rho (1-4 \alpha \theta)$ and $(1+\beta) \rho /(1 - 4\alpha) \le (1+\epsilon)$. 
Thus we obtain $(1-\epsilon) f_k \le \hat{f}_k \le (1+ \epsilon) f_k$. 

\ignore{
	\begin{lemma}
	\label{lem:finalnofp}
	Using an estimate $\hat{F}_0$ with $|\hat{F}_0 - F_0| \le \epsilon$
	The level $w^*$ computed by the algorithm is such that $|Y_{w^*} - E(X_{w^*})| \le \epsilon E(X_{w^*})$ with probability at least $1/4$ when $r = O(\lambda/\epsilon^2)$ and $u= O(\lambda\log(F_0)/\epsilon^2)$.
	\end{lemma}
}


\vspace*{0.2cm}
\noindent
\textbf{$F_0$ computation}
\vspace*{0.2cm}

Estimate for $F_0$ can be computed by any of the existing approaches. Our algorithm uses the same approach of \cite{kmerstream} to compute $F_0$. It is straightforward to verify from the proof of Theorem $1$ in \cite{kmerstream} that for $r = O(1/\epsilon^2)$, the estimate $\hat{p}_0$ of  $p_0 = (1 - 1/r)^{F_0/2^{w^*}}$ computed in our $F_0$ estimation procedure, where $w^{*}$ is the level chosen by the algorithm, has the property that
$|\hat{p}_0 - p_0| > \frac{\epsilon}{3} p_0$ with probability at most $1/2$. 
Furthermore, we recall from  the proof of Theorem 1 \cite{kmerstream} that $p_0 \ge 1/3$.

It follows that to estimate $\hat{F}_0$ such that  $\hat{F}_0 \in (1 \pm \alpha)F_0$, $\alpha$ and $\epsilon$ are related as 
$$
1 - \epsilon/3 ~\le~  p_0^\alpha ~~\mbox{and}~~  1 + \epsilon/3  ~\ge~  p_0^{-\alpha}
$$
It is straightforward to verify that the above conditions are satisfied when $\epsilon/3 \le \alpha \le \epsilon$.
It follows that for $r = O(1/\alpha^2)$, estimate $\hat{F}_0$ of $F_0$ can be computed such that  $|\hat{F}_0 - F_0| > \alpha F_0$ with probability at most  $1/8$.
Using standard Chernoff bound, it follows that for $s = O(\log(1/\delta))$, more than half of $s$ independent estimates of $F_0$ deviate by more than $\alpha F_0$ from $F_0$ is at most $\delta/2$. In other words, the median of $s$ independent estimates deviates by more than $\alpha F_0$ from $F_0$ with probability at most $\delta/2$.

\vspace*{0.2cm}
\noindent
\textbf {Probability amplification through median computation}
\vspace*{0.2cm}

Using again the median argument given above, it follows that by using median of $O(\log(\lambda/\delta))$
independent estimates as the final estimate, for any fixed $f_k$ for $k \ge 1$, an estimate $\hat{f}_{k}$ can be computed such that $|\hat{f}_k - f_k| \ge \epsilon f_k$ with probability at most $\delta/(2\lambda)$. Noting that there are at most $\lambda$ such $k$ with $f_k \ge F_0 / \lambda$, it follows by union bound that the probability of obtaining a bad estimate for either $F_0$ or for any $f_k$  where $f_k \ge F_0 / \lambda$, is at most $\delta$.

\vspace*{0.2cm}
\noindent
\textbf {Space and Time Complexities}
\vspace*{0.2cm}

We recall that each sketch instance has $O(\log(F_0))$ levels where each level contain an array of $r$ counters. 
Using  $O(\frac{\lambda}{\epsilon^2} \log(1/\delta) \log(F_0))$ memory  and $O(\log(1/\delta))$ update time,
the algorithm computes estimate $\hat{f}_k$ for any fixed  $k \ge 1$ and with $f_k \ge F_0/\lambda$, such that 
$|\hat{f}_k - f_k| \ge \epsilon f_k$ with probability at most $\delta$. 
Each memory word has $O(\log(k) + \log(\lambda F_0/\epsilon^2))$ bits.
Using $O(\frac{\lambda}{\epsilon^2} \log(\lambda/\delta) \log(F_0))$ memory and $O(\log(\lambda/\delta))$ update time,
with probability at least $1 - \delta$, estimate $\hat{F}_0$ for $F_0$ and estimates $\hat{f}_i$ for every $f_i$ with $f_i \ge F_0/\lambda$ can be computed 
such that $(1-\epsilon) F_0 \le \hat{F}_0 \le (1+\epsilon) F_0$ and  $(1-\epsilon) f_i \le \hat{f}_i \le (1+\epsilon) f_i$.
Each memory word in this case require $O(\log(k) + \log(\lambda F_0/\epsilon^2))$ bits where $k = \arg\max_i\{f_i \ge F_0/\lambda\}$.

The proof presented here does not attempt to obtain tight constants in the asymptotic space bounds. A more rigorous proof could possibly provide space bounds with tight constants that matches the experimental findings.

\ignore{

	\vspace*{2cm}
	\noindent
	\textbf{Old material}

	Fix any $F_j$ such that $F_j \ge F_0 / \lambda$.  In order to compute an estimate $\hat{F}_j$ of $F_j$, we will use a good estimate $\hat{F}_0$ of $F_0$ as given by the following claim whose proof is given later.

	\begin{claim}
	An estimate $\hat{F}_0$ of $F$ known where $|\hat{F}_0 - F_0| \le \epsilon F_0$ with probability at least $1/\delta_1$.
	\end{claim}

	We will also use the following claim whose proof is given later. We will consider an interval of rows $[l, u]$ in order to estimate $F_j$. Values of $l$ and $r$ will be fixed later.
	Let the interval $[l, u]$ be chosen such that any $w \in [l, u]$ satisfies the property that $r/3 \le F_0/2^w \le 2r/3$. 

	\begin{claim}
	Fix any row $w \in [l, u]$. Let $N_0$ denote the number of distinct elements hashed to row $w$ and let $N_j$ denote the number of distinct elements each with frequency $j$ that are hashed to row $w$. Then
	\begin{eqnarray*}
	\Pr[|N_0 - F_0/2^w| \ge r/12 ] \le \beta_2\\
	\Pr[|2^w \cdot N_j - F_j| \ge \epsilon F_j]  \le 1/\beta_3
	\end{eqnarray*}
	\end{claim}

	\begin{lemma}
	Consider any fixed row $w \in [l, u]$. Let random variable $X$ denote the number of non-dirty array locations with value $j$ in row $w$.  Then $|X - E(X)| \le \epsilon E(X)$ with probability at least $1- \beta_4$ for $r = O(\lambda^2/\epsilon^2)$.
	\end{lemma}

	\begin{proof}
	We will assume the range of $N_0$ and $N_j$ as given in the above claim. Combining with the assumption that $r/3 \le F_0/2^w \le 2r/3$, this implies
	\begin{eqnarray}
	\label{one}
	r/4 \le N_0 \le 3r/4\\
	\label{two}
	\left(1 - \frac{1}{r}\right)^{N_0} \ge 1/3\\
	\label{three}
	 \frac{r(1-\epsilon)}{3\lambda} \le N_j \le \frac{2r(1+\epsilon)}{3\lambda}
	\end{eqnarray}

	 In our algorithm, after the first level hashing into a row, there is a second level hashing that assigns the elements into array locations. This hashing is independent of the first level hashing.

	Let $p_j$ denote the probability that a fixed array location $k$ in row $w$ is occupied by exactly one element from the $N_j$ elements and no other elements. This is given by
	$$
	p_j = \frac{N_j}{r} \left( 1 - \frac{1}{r}\right)^{N_0 - 1}
	$$
	Hence, $E(X) = rp_j$. 
	Also, we note that $X$ is 1-Lipschitsz as changing the location of any one of the element $N_0$ elements can alter the value of $X$ by at most one. Hence by Azuma-H\"offding inequality, we get

	\begin{eqnarray*}
		\Pr[|X - E(X)| \ge \epsilon E(X)] & \le  & 2\cdot\exp\left(-\frac{(\epsilon E(X))^2}{2N_0} \right)\\
						& =  & 2\cdot\exp\left(-\frac{(\epsilon rp_j)^2}{2N_0} \right)
	\end{eqnarray*}

	We use the facts $F_j \ge F_0/\lambda$ and (\ref{one}), (\ref{two}) and (\ref{three}), we obtain that
	\begin{eqnarray*}
		\Pr[|X - E(X)| \ge \epsilon E(X)] & \le  & 2\cdot\exp\left(-\frac{(\epsilon N_j ( 1 - 1/r)^{N_0-1})^2}{2N_0} \right)\\
		& =  & 2\cdot\exp\left(-\frac{(\epsilon N_j ( 1 - 1/r)^{N_0})^2}{2N_0(1 - 1/r)} \right)\\
		& =  & 2\cdot\exp\left(-\frac{\epsilon^2 r (1-\epsilon)^2}{243 \cdot \lambda^2 (1- 1/r)}\right) \\
		& \le \beta_4
	\end{eqnarray*}
	for $r = O(\lambda^2/\epsilon^2)$.

	\end{proof}

}



\ignore{

}
\end{document}